\documentclass[sigconf,nonacm]{aamas} 

\usepackage{subfig}
\usepackage{algorithm,algpseudocode}

\usepackage{amsmath}
\usepackage{amsthm}
\usepackage{booktabs}
\usepackage{mathtools}

\DeclarePairedDelimiter\floor{\lfloor}{\rfloor}
\usepackage{tabularx}
\newcolumntype{L}[1]{>{\raggedright\arraybackslash}m{#1}}
\newcolumntype{P}[1]{>{\centering\arraybackslash}p{#1}}
\newtheorem{theorem}{Theorem}
\usepackage{enumitem}
\usepackage{balance} % for balancing columns on the final page

\thanks{This is the author’s version of the work. It is posted here for your personal use only. An extended abstract version of this work appears in the \textit{Proc.\@ of the 25th International Conference on Autonomous Agents and Multiagent Systems (AAMAS 2026)} \cite{Altmann26-DRIVE}.
}

\title{Dynamic Incentivized Cooperation under Changing Rewards}

\author{Philipp Altmann}
\affiliation{\institution{LMU Munich, Germany}\country{}\city{}}
\email{philipp.altmann@ifi.lmu.de}

\author{Thomy Phan}
\affiliation{\institution{University of Bayreuth, Germany}\country{}\city{}}

\author{Maximilian Zorn}
\affiliation{\institution{LMU Munich, Germany}\country{}\city{}}

\author{Claudia Linnhoff-Popien}
\affiliation{\institution{LMU Munich, Germany}\country{}\city{}}

\author{Sven Koenig}
\affiliation{\institution{University of California, Irvine, USA}\country{}\city{}}
\affiliation{\institution{Orebro University, Sweden}\country{}\city{}}

\begin{abstract}
\emph{Peer incentivization} (PI) is a popular multi-agent reinforcement learning approach where all agents can reward or penalize each other to achieve cooperation in social dilemmas. 
Despite their potential for scalable cooperation, current PI methods heavily depend on fixed incentive values that need to be appropriately chosen with respect to the environmental rewards and thus are highly sensitive to their changes.
Therefore, they fail to maintain cooperation under \emph{changing rewards} in the environment, e.g., caused by modified specifications, varying supply and demand, or sensory flaws --- even when the conditions for mutual cooperation remain the same.
In this paper, we propose \emph{Dynamic Reward Incentives for Variable Exchange} (DRIVE), an adaptive PI approach to cooperation in social dilemmas with changing rewards. DRIVE agents reciprocally exchange reward differences to incentivize mutual cooperation in a completely decentralized way. 
We show how DRIVE achieves mutual cooperation in the general Prisoner's Dilemma and empirically evaluate DRIVE in more complex sequential social dilemmas with changing rewards, demonstrating its ability to achieve and maintain cooperation, in contrast to current state-of-the-art PI methods.
\\[4pt]
\textbf{Code:} \url{https://github.com/philippaltmann/DRIVE}
\end{abstract}


\begin{document}

\pagestyle{fancy}
\fancyhead{}

%%% The next command prints the information defined in the preamble.

\maketitle

%%%%%%%%%%%%%%%%%%%%%%%%%%%%%%%%%%%%%%%%%%%%%%%%%%%%%%%%%%%%%%%%%%%%%%%%

\section{Introduction}

Many AI scenarios, such as autonomous driving \cite{shalev2016safe}, smart grids \cite{dimeas2010multi}, and IoT applications \cite{deng2020dynamical}, can be modeled as self-interested and online learning \emph{multi-agent systems (MAS)}, where conflicts arise due to opposing goals or shared resources \cite{bucsoniu2010multi}.
To maximize \emph{social welfare}, cooperation is essential in self-interested MAS, requiring agents to avoid selfish behavior for the collective good. This tension between individual and collective rationality is typically modeled as a \emph{social dilemma (SD)} \cite{rapoport1974prisoner,axelrod1984evolution}. Designing distributed mechanisms to incentivize cooperation in SDs remains a key challenge because it directly affects global efficiency (e.g., aggregated social welfare as defined in Sec.~\ref{subsec:problem_formulation}) and sustainability (e.g., long-run resource availability and fairness across agents) \cite{axelrod1984evolution,trivers1971evolution,dafoe2020open,perolat2017learning}.
\emph{Multi-agent reinforcement learning (MARL)} is a widely used framework to train rational agents in SDs and temporally extended \emph{sequential SDs (SSD)}, where each agent maximizes its own reward \cite{leibo2017multi,perolat2017learning,foerster2018learning,yang2020learning}. 
\begin{figure}\centering
\includegraphics[width=0.96\linewidth]{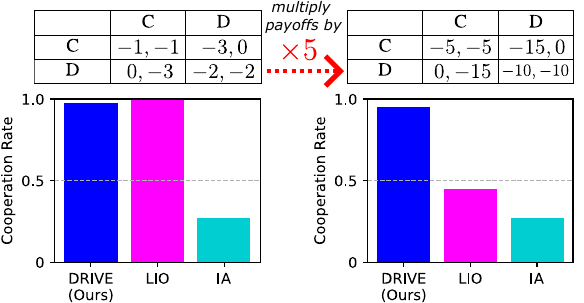}
\caption{Motivating example on the reward sensitivity of PI approaches in a 2-player Prisoner’s Dilemma: IA requires careful tuning even under the original payoffs and fails to achieve cooperation in both scenarios. LIO achieves high cooperation under the original scale (left) but degenerates when the payoffs change, even though the inequalities for greed and fear (Eq.~\ref{eq:PD_inequalities}) still hold. In contrast, DRIVE maintains robust cooperation across both settings without retuning.}\label{fig:intro_figure}
\Description{Motivating example on the reward sensitivity of PI approaches in a 2-player Prisoner’s Dilemma: IA requires careful tuning even under the original payoffs and fails to achieve cooperation in both scenarios. LIO achieves high cooperation under the original scale (left) but degenerates when the payoffs change, even though the inequalities for greed and fear (Eq.~\ref{eq:PD_inequalities}) still hold. In contrast, DRIVE maintains robust cooperation across both settings without retuning.}
\end{figure}
% \\[2pt]
\newpage
However, non-cooperative game theory and empirical studies show that naive MARL approaches often fail to sustain cooperation due to independent adaptation leading to mutual defection \cite{axelrod1984evolution,van2013psychology}.
Further, policy optimization can be affected by inconsistencies in the reward function, summarized as \emph{changing rewards}. 
Such variations are common when moving from simulation to reality, where abstract reward proxies are replaced by physical signals. 
Noisy sensors, hardware degradation, shifting specifications, or changing market demands can likewise alter reward scales during training \cite{dulac2019challenges}. 
In human-in-the-loop settings, evaluators may also adapt feedback online, leading to abrupt or irregular shifts \cite{knox2009interactively}. 
In these cases, the underlying SSD structure (i.e., greed and fear inequalities) is preserved, but the numerical values fluctuate. 
Methods that are brittle to such changes struggle to sustain cooperation beyond narrow, hand-tuned conditions.
\\[2pt]
\emph{Peer incentivization} (PI) enables agents to reward or penalize each other to foster cooperation in SDs \cite{lupu2020gifting,yang2020learning}. It has gained traction for its success in complex SSDs and links to biology, economics, and social science \cite{vinitsky2021learning,schmid2021stochastic,koelle23ltp}. Yet, existing PI methods rely on fixed or domain-tuned incentive values and fail under changing rewards, even if cooperation conditions remain unchanged \cite{eysenbach2021maximum,zhang2020adaptive}.
As Fig.~\ref{fig:intro_figure} illustrates, prior PI methods might succeed in a given instance of the Prisoner’s Dilemma \cite{foerster2018learning,phanJAAMAS2024,yang2020learning}, but fail when payoffs shift or scale, despite satisfying the same \emph{greed} and \emph{fear} inequalities \cite{axelrod1984evolution,macy2002learning}. % % 
Notably, popular approaches such as \emph{Learning to Incentivize Other learning agents} (LIO) \cite{yang2020learning} or \textit{inequity aversion} (IA) \cite{hughes2018inequity} appear superficially capable of adapting to such variation, since they rely on learned or relative incentives. 
Yet, IA requires precise hyperparameter tuning to be effective and fails to promote cooperation in both the original (unscaled) payoff matrix of Fig.~\ref{fig:intro_figure} (left) and the uniformly scaled variant (right).
LIO achieves cooperation under the original payoffs but degrades under simple scaling, as its learned incentive magnitudes are tied to the absolute reward scale.
As a result, these methods either lose robustness under reward shifts or suppress valuable heterogeneity in agent performance. 
Addressing such shifts typically requires hyperparameter retuning, which is impractical for ad-hoc- or online learning scenarios.
\\[2pt]
To this end, we propose \emph{Dynamic Reward Incentives for Variable Exchange (DRIVE)}, an adaptive PI framework for SDs with changing rewards. 
By adapting incentives directly to \emph{reward differences} DRIVE sustains cooperation without retuning and without conflating legitimate variance with exploitative behavior.
We summarize our contributions as follows:
\\[2pt]
\begin{itemize}[left=4pt]
\setlength{\itemsep}{4pt}\setlength{\topsep}{20pt}
\item We introduce a reciprocal exchange mechanism to incentivize cooperation via reward differences.
\item We prove that DRIVE aligns incentives toward mutual cooperation in generalized Prisoner’s Dilemma games, remains invariant to reward shifts and scaling. 
\item We empirically demonstrate DRIVE’s robustness and superior cooperation performance in SSDs with changing rewards.
\end{itemize}
\vspace{8pt}

\section{Background}
\noindent\textbf{Problem Formulation. }\label{subsec:problem_formulation}
In this work, we focus on partially observable \emph{Markov games} $M = \langle \mathbb{D},\mathbb{S},\mathbb{A},\mathbb{P},\mathbb{U},\mathbb{Z},\Omega \rangle$, where $\mathbb{D} = \{1,...,n\}$ is a set of agents $i$, $\mathbb{S}$ is a set of states $s_{t}$ at time step $t$, $\mathbb{A} = \langle \mathbb{A}_{1}, ..., \mathbb{A}_{n} \rangle =  \langle \mathbb{A}_{i} \rangle_{i \in \mathbb{D}}$ is the set of joint actions $a_{t} = \langle a_{t,i} \rangle_{i \in \mathbb{D}}$, $\mathbb{P}(s_{t+1}|s_{t}, a_{t})$ is the transition probability, $\langle u_{t,i} \rangle_{i \in \mathbb{D}} = \mathbb{U}(s_{t},a_{t}) \in \mathbb{R}$ is the joint reward, $\mathbb{Z}$ is a set of local observations $z_{t,i}$ for each agent $i$, and $\Omega(s_{t+1}) = z_{t+1} = \langle z_{t+1,i} \rangle_{i \in \mathbb{D}} \in \mathbb{Z}^{n}$ is the subsequent joint observation. Each agent $i$ maintains a local \emph{history} $\tau_{t,i} \in (\mathbb{Z} \times \mathbb{A}_{i})^{t}$. $\pi_{i}(a_{t,i}|\tau_{t,i})$ is the action selection probability based on the individual \emph{policy} of agent $i$. In addition, we assume each agent $i$ to have a \emph{neighborhood} $\mathcal{N}_{t,i} \subseteq \mathbb{D} - \{i\}$ of other agents at every time step $t$, which is domain-dependent, as suggested in \cite{lupu2020gifting,yang2018mean}.
$\pi_{i}$ is evaluated with a \emph{value function} $V_{i}^{\pi}(s_{t}) = \mathbb{E}_{\pi}[G_{t,i}|s_{t}]$ for all $s_{t} \in \mathbb{S}$, where $G_{t,i} = \sum_{k=0}^{\infty} \gamma^{k} u_{t+k,i}$ is the individual and discounted \emph{return} of agent $i \in \mathbb{D}$ with discount factor $\gamma \in [0,1)$ and $\pi = \langle \pi_{j} \rangle_{j \in \mathbb{D}}$ is the \emph{joint policy} of the MAS. The goal of agent $i$ is to find a \emph{best response} $\pi_{i}^{*}$ with $V_{i}^{*} = max_{\pi_{i}}V_{i}^{\langle \pi_{i}, \pi_{-i} \rangle}$ for all $s_{t} \in \mathbb{S}$, where $\pi_{-i}$ is the joint policy \emph{without} agent $i$. 
In practice, the global state $s_{t}$ is not directly observable for any agent $i$ s.t. $V_{i}^{\pi}$ is approximated with local information, i.e., $\tau_{t,i}$ instead \cite{leibo2017multi,perolat2017multi,jaderberg2019human}.
To measure cooperation in the MAS, we define the \emph{social welfare} or \emph{utilitarian metric} $U = \sum_{i \in \mathbb{D}}\sum_{t=0}^{H-1} u_{t,i}$ as undiscounted sum of rewards.
\newpage

\noindent\textbf{Social Dilemmas. }\label{subsec:social_dilemmas}
\emph{Social dilemmas (SD)} are games where independently optimized policies $\pi_{i}$ fail to achieve globally optimal outcomes that maximize collective welfare. SDs are commonly studied in 2-player matrix games with two actions: $C$ (cooperate) and $D$ (defect), producing four possible payoffs (Fig.~\ref{fig:social_dilemma_matrix}): $R$ for mutual cooperation, $P$ for mutual defection, $T$ for exploiting the other, and $S$ for being exploited. A matrix game is a \emph{Prisoner's Dilemma (PD)} if the payoffs satisfy \cite{axelrod1984evolution,macy2002learning}: 
\begin{equation}\label{eq:PD_inequalities}
T > R > P > S
\end{equation}
Here, $T > R$ represents \emph{greed}, and $P > S$ \emph{fear}. In \emph{Iterated PDs (IPD)}, an additional condition holds: $2R > T + S$ \cite{axelrod1984evolution,macy2002learning}. Fig.~\ref{fig:PD_instance} shows an instance where $D$ is individually rational despite $C$ being socially optimal. As long as the inequalities are satisfied, the game's strategic nature remains invariant to exact payoffs \cite{axelrod1984evolution,rapoport1974prisoner}.
PDs and IPDs are particularly important SDs, as greed and fear often drive agents away from mutual cooperation despite its collective benefit, a phenomenon observed in both nature and human society \cite{axelrod1984evolution,dawkins2016selfish,rapoport1974prisoner}.
\\[2pt]
\emph{Sequential social dilemmas (SSD)} extend SDs by introducing temporal structure, modeled as stochastic games \cite{leibo2017multi,perolat2017learning}. SSDs allow more realistic scenarios where behavior is captured by policies, not atomic actions. These can still be mapped to matrix games by classifying policies as $C$ or $D$ and evaluating empirical payoffs \cite{leibo2017multi}, making core SD concepts applicable to SSDs.
\\[2pt]
\begin{figure}\hfill
\subfloat[Social dilemma payoffs\label{fig:canonical_game}]{
\begin{tabular}[b]{|P{3em}|P{3em}|P{3em}|}
\hline & $C$ & $D$ \\ \hline 
$C$ & $R$, $R$ & $S$, $T$ \\ \hline 
$D$ & $T$, $S$ & $P$, $P$ \\ \hline
\end{tabular}}\hfill
\subfloat[PD instance\label{fig:PD_instance}]{%
\begin{tabular}[b]{|P{3em}|P{3em}|P{3em}|}
\hline & $C$ & $D$ \\ \hline 
$C$ & $-1$, $-1$ & $-3$, $0$ \\ \hline 
$D$ & $0$, $-3$ & $-2$, $-2$ \\ \hline
\end{tabular}}\hfill
\caption{\textbf{(a)} Social dilemma payoff matrix with $R$, $P$, $T$, and $S$. In \emph{Prisoner's Dilemmas (PD)}, $T > R > P > S$ holds \cite{axelrod1984evolution,macy2002learning}. \textbf{(b)} A PD instance satisfying these inequalities.}\label{fig:social_dilemma_matrix}
\Description{\textbf{(a)} Social dilemma payoff matrix with $R$, $P$, $T$, and $S$. In \emph{Prisoner's Dilemmas (PD)}, $T > R > P > S$ holds \cite{axelrod1984evolution,macy2002learning}. \textbf{(b)} A PD instance satisfying these inequalities.}
\end{figure}
\noindent\textbf{Multi-Agent Reinforcement Learning. }\label{subsec:marl_background}
We consider decentralized (independent) learning, where each agent $i$ optimizes its policy $\pi_{i}$ using local data like $\tau_{t,i}$, $a_{t,i}$, $u_{t,i}$, and $z_{t+1,i}$ via \emph{reinforcement learning (RL)} \cite{tan1993multi,foerster2018learning,yang2018mean}. 
For methods with peer incentives, the shaped reward simply replaces $u_{t,i}$ in this trajectory.
\emph{Policy gradient RL} is a common method to approximate best responses $\pi_{i}^{*}$ \cite{lowe2017multi,foerster2018learning,yang2020learning}. 
A function approximator $\hat{\pi}_{i,\phi_{i}} \approx \pi_{i}^{*}$ is trained using gradient ascent on an estimate of $J = \mathbb{E}_{\pi}[G_{0,i}]$ \cite{williams1992simple}, where the policy gradient is approximated as \cite{sutton2000policy}:
\begin{equation}\label{eq:policy_gradients}
g = (G_{t,i} - b_{i}(s_{t}))\nabla_{\phi{i}} \textit{log} \hat{\pi}_{i,\phi_{i}}(a_{t,i}|\tau_{t,i})
\end{equation}
Here, $b_{i}(s_{t})$ is a state-dependent \emph{baseline}, typically approximated by a learned value function $\hat{V}_{i,\omega_{i}}(\tau_{t,i}) \approx V_{i}^{\hat{\pi}}(s_{t})$ \cite{foerster2018learning}. For simplicity, we omit parameters and write $\hat{\pi}_{i}$, $\hat{V}_{i}$.
Modern actor-critic methods add such variance-reduction baselines, entropy bonuses, or centralized critics, but still only optimize each agent's own return, which remains misaligned with social welfare in social dilemmas. Also, \emph{independent learning} introduces non-stationarity as agents adapt simultaneously \cite{laurent2011world,hernandez2017survey}, often driving overly greedy behavior and mutual defection unless rewards or incentives are modified.

\section{Dynamic Reward Incentivization}
We assume a decentralized MARL setting as formulated in Algorithm \ref{algorithm:MARL_DRIVE}, where at every time step $t$ 
each agent $i$ with history $\tau_{t,i}$, policy $\hat{\pi}_{i}$, and value function $\hat{V}_{i}$ observes its neighborhood $\mathcal{N}_{t,i}$ through its local observation $z_{t,i}$
and executes an action $a_{t,i} \sim\pi_{i}(\cdot|\tau_{t,i})$. 
DRIVE uses $\mathcal{N}_{t,i}$ only for the incentive exchange; the underlying policy does not receive additional messages beyond the environment observation.
The environment then transitions to a new state $s_{t+1} \sim \mathbb{P}(\cdot|s_{t}, a_{t})$ which is observed by each agent $i$ through an observation $z_{t+1,i}$ and a reward $\hat{u}_{t,i}$.
This is obtained by passing the environmental reward $u_{t,i}$ through an external and possibly unknown \emph{reward‑change function} $f_{\textit{mod}}$ (Alg.~\ref{algorithm:MARL_DRIVE}, l.~\ref{l:r_mod}), which simulates modified specifications, varying supply and demand, or sensor degradation \cite{eysenbach2021maximum,zhang2020adaptive}. 
In Sec.~\ref{sec:theoretical_analysis} we analyze an affine special case of $f_{\textit{mod}}$, applying a shared map varying by epoch.
All agents collect their respective \emph{experience tuple} $e_{t,i} = \langle \tau_{t,i}, a_{t,i}, \hat{u}_{t,i}, z_{t+1,i} \rangle$ for PI exchange and independent updates to $\hat{\pi}_{i}$ and $\hat{V}_{i}$ \cite{yang2020learning}. 
We assume each agent $i$ has a domain‑dependent neighborhood $\mathcal{N}_{t,i}\subseteq D \setminus \{i\}$ at each time step, as introduced in Sec.~\ref{subsec:problem_formulation}. Lines 13–20 in Alg.~\ref{algorithm:MARL_DRIVE} correspond to a short parallel communication phase in which neighboring agents exchange DRIVE requests and responses.
\begin{algorithm} % Alg1
\caption{Multi-Agent Learning with DRIVE}\label{algorithm:MARL_DRIVE}
\begin{algorithmic}[1]
\State Initialize parameters of $\hat{\pi}_{i}$ and $\hat{V}_{i}$ for all agents $i \in \mathbb{D}$
\For{epoch $m \leftarrow 1,E$}
\State Set $\overline{u}_i \leftarrow 0$ for all agents $i \in \mathbb{D}$ \Comment{Reset avg. reward}
\State Sample $s_{0}$ and set $\tau_{0,i}$ for all agents $i \in \mathbb{D}$
\For{time step $t \leftarrow 0,H-1$}
\For{agent $i \in \mathbb{D}$} \Comment{Independent decisions}
\State Observe local neighborhood $\mathcal{N}_{t,i}$
\State $a_{t,i} \sim \hat{\pi}_{i}(\cdot\text{\textbar}\tau_{t,i})$
\EndFor
\State $s_{t+1} \sim \mathbb{P}(\cdot\text{\textbar}s_{t}, a_{t})$ \Comment{Execute joint action $a_{t}$}
\State $\langle u_{t,i} \rangle_{i \in \mathbb{D}} \leftarrow \mathbb{U}(s_{t},a_{t})$ \Comment{Environmental rewards}
\State $\langle z_{t+1,i} \rangle_{i \in \mathbb{D}} \leftarrow \Omega(s_{t+1})$
\For{agent $i \in \mathbb{D}$} \Comment{Parallel communication}
\State $\hat{u}_{t,i} \leftarrow f_{\textit{mod}}(u_{t,i}, m)$ \Comment{Ext. reward change}\label{l:r_mod}
\State $\overline{u}_{i} \leftarrow (t\cdot\overline{u}_{i} + \hat{u}_{t,i})/(t+1)$ \Comment{Avg. reward}
\State $u^{\textit{DRIVE}}_{t,i} \leftarrow \textit{DRIVE}(\hat{V}_{i}, \mathcal{N}_{t,i}, \hat{u}_{t,i}, \overline{u}_{i})$ \Comment{Alg. \ref{algorithm:DRIVE}}
\State $e_{t,i} \leftarrow \langle \tau_{t,i}, a_{t,i}, u^{\textit{DRIVE}}_{t,i},  z_{t+1,i} \rangle$
\State Update $\tau_{t,i}$ to $\tau_{t+1,i}$ and store $e_{t,i}$
\EndFor
\EndFor
\For{agent $i \in \mathbb{D}$} \Comment{Independent updates}
\State Update $\hat{\pi}_{i}$ and $\hat{V}_{i}$ using all $e_{t,i}$ of epoch $m$
\EndFor
\EndFor
\end{algorithmic}
\end{algorithm}
\\[2pt]
\textbf{DRIVE Token Exchange. }
DRIVE uses a reciprocal incentive scheme, inspired by \citet{trivers1971evolution} and illustrated in Fig. \ref{fig:drive_protocol}, to exchange dynamic incentives for distributed reward shaping.
In the \emph{request phase} (Fig. \ref{fig:drive_request}), each agent $i$ checks its \emph{advantage} or \emph{temporal difference residual} $\textit{TD}_{i}(\hat{u}_{t,i})$ \cite{sutton1988learning}:
\begin{equation}\label{eq:mi_td}
\textit{TD}_{i}(\hat{u}_{t,i}) = \hat{u}_{t,i} + \gamma \hat{V}_{i}(\tau_{t+1,i}) - \hat{V}_{i}(\tau_{t,i})
\end{equation}
If the advantage is non-negative, agent $i$ sends its reward $\hat{u}_{t,i}$ as a \emph{request} to all other agents $j \in \mathcal{N}_{t,i}$.
In the \emph{response phase} (Fig. \ref{fig:mate_response}), each request‑receiving neighbor $j \in \mathcal{N}_{t,i}$ compares the request $\hat{u}_{t,i}$ to its epoch‑average reward $\overline{u}_j$ and computes $\Delta_{t,1,2} = \overline{u}_2 - \hat{u}_{t,1}$, which is then sent back as a response.
After this exchange, the DRIVE reward $u^{\textit{DRIVE}}_{t,i}$ is computed for each agent $i$ as follows:
\begin{equation}\label{eq:DRIVE_reward}
u^{\textit{DRIVE}}_{t,i} = \hat{u}_{t,i} - \textit{min}\{\langle \Delta_{t,j,i} \rangle_{j \in \mathcal{N}_{t,i}}\} + \textit{min}\{\langle \Delta_{t,i,j} \rangle_{j \in \mathcal{N}_{t,i}}\}
\end{equation}
Using the responder’s epoch‑average $\overline{u}_j$ rather than its instantaneous reward makes DRIVE sensitive to systematic exploitation instead of single noisy outcomes: only agents whose recent average return is consistently lower than a neighbor’s request generate strong negative $\Delta$ terms. Other combinations (instant–instant, instant–average, average–average) are conceivable and may yield different trade‑offs, which is an interesting direction for future work.
Appendix~\ref{app:detailed_PD} provides an initial comparison.
\begin{figure}[t]\hfill
\subfloat[DRIVE request\label{fig:drive_request}]{\includegraphics[width=0.2\textwidth]{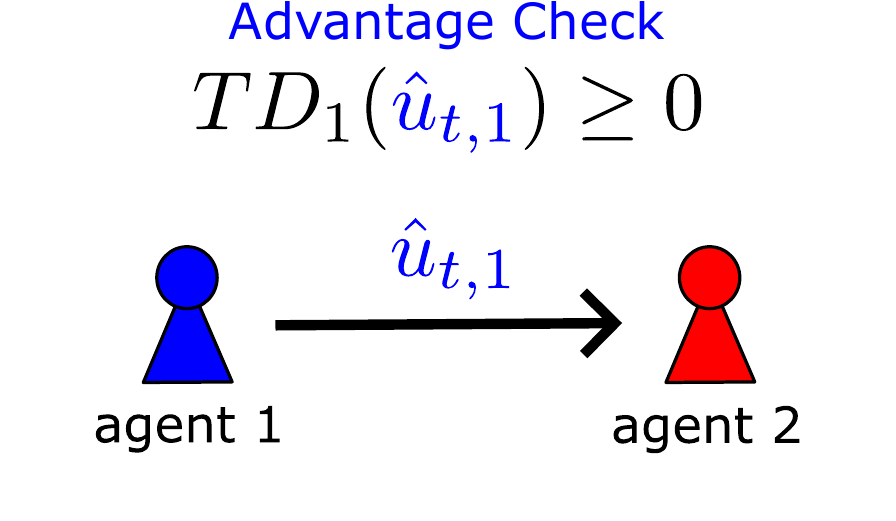}}\hfill
\subfloat[DRIVE response\label{fig:mate_response}]{\includegraphics[width=0.2\textwidth]{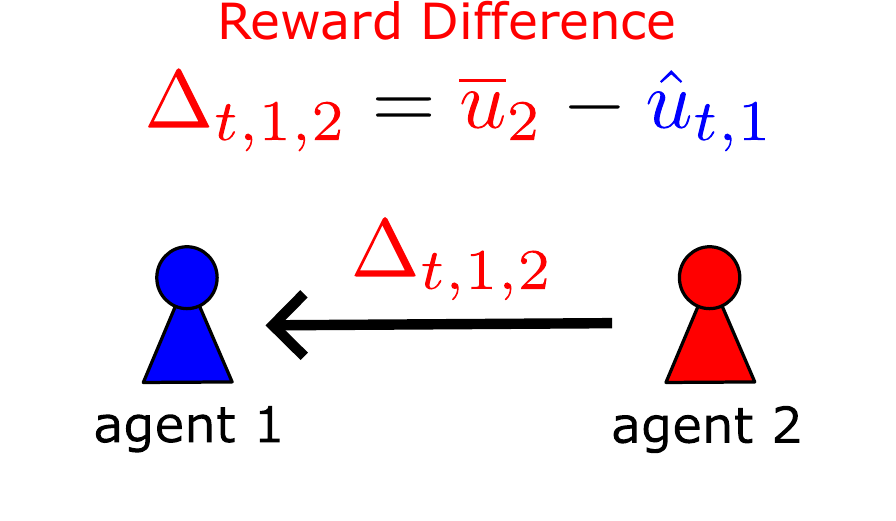}}\hfill
\caption{DRIVE exchange scheme. \textbf{(a)} If $\textit{TD}_1(\hat{u}_{t,1}) \geq 0$ (Eq.~\ref{eq:mi_td}), agent 1 sends its reward $\hat{u}_{t,1}$ to neighbor agent 2 as a request. \textbf{(b)} Agent 2 calculates the difference $\Delta_{t,1,2}$ between its own average reward $\overline{u}_i$ in the current epoch $m$ and the received request $\hat{u}_{t,i}$, sent back as a response and used to shape the rewards of both agents (Eq.~\ref{eq:DRIVE_reward}).}\label{fig:drive_protocol}
\Description{DRIVE exchange scheme. \textbf{(a)} If $\textit{TD}_1(\hat{u}_{t,1}) \geq 0$ (Eq.~\ref{eq:mi_td}), agent 1 sends its reward $\hat{u}_{t,1}$ to neighbor agent 2 as a request. \textbf{(b)} Agent 2 calculates the difference $\Delta_{t,1,2}$ between its own average reward $\overline{u}_i$ in the current epoch $m$ and the received request $\hat{u}_{t,i}$, sent back as a response and used to shape the rewards of both agents (Eq.~\ref{eq:DRIVE_reward}).}
\end{figure}
\begin{algorithm} % Alg 2
\caption{Dynamic Reward Incentive Exchange}\label{algorithm:DRIVE}
\begin{algorithmic}[1]
\Procedure{$\textit{DRIVE}(\hat{V}_{i}, \mathcal{N}_{t,i}, \hat{u}_{t,i}, \overline{u}_{i})$}{}
\State $\hat{u}_{\textit{req}} \leftarrow \infty$, $\hat{u}_{\textit{res}} \leftarrow 0$
\If{$\textit{TD}_{i}(\hat{u}_{t,i}) \geq 0$} \Comment{Evaluate acc. to Eq.~\eqref{eq:mi_td}}
\State Send request $\hat{u}_{t,i}$ to all $j \in \mathcal{N}_{t,i}$ (Fig. \ref{fig:drive_protocol}a)
\EndIf
\For{neighbor agent $j \in \mathcal{N}_{t,i}$}\Comment{Requests}
\If{request $\hat{u}_{t,j}$ received from $j$}
\State $\Delta_{t,j,i} \leftarrow \overline{u}_{i} - \hat{u}_{t,j}$
\State $\hat{u}_{\textit{req}} \leftarrow \textit{min}\{\hat{u}_{\textit{req}}, \Delta_{t,j,i}\}$
\State Send response $\Delta_{t,j,i}$ to agent $j$ (Fig. \ref{fig:drive_protocol}b)
\EndIf
\EndFor
\State \textbf{if} $\hat{u}_{\textit{req}} = \infty$ \textbf{then} $\hat{u}_{\textit{req}} \gets 0$ \Comment{If no requests received}
\If{$\textit{TD}_{i}(\hat{u}_{t,i}) \geq 0$}\Comment{If requests have been sent}
\State $\hat{u}_{\textit{res}} \leftarrow \infty$
\For{neighbor agent $j \in \mathcal{N}_{t,i}$}\Comment{Responses}
\State $\hat{u}_{\textit{res}} \leftarrow \textit{min}\{\hat{u}_{\textit{res}}, \Delta_{t,i,j}\}$
\EndFor
\EndIf
\State \Return $\hat{u}_{t,i} - \hat{u}_{\textit{req}} + \hat{u}_{\textit{res}}$ \Comment{Shaped reward (Eq.~\ref{eq:DRIVE_reward})}
\EndProcedure
\end{algorithmic}
\end{algorithm}
\\[2pt]\textbf{Distributed Reward Shaping. }
The DRIVE reward $u^{\textit{DRIVE}}_{t,i}$ is used to update the policies of the corresponding agents using any RL algorithm, e.g., policy gradient methods, according to Eq.~\eqref{eq:policy_gradients}.
Formally, whenever DRIVE is enabled we obtain returns $G_{t,i}$ by
replacing $u_{t,i}$ with $u^{\text{DRIVE}}_{t,i}$ in Eq.~\eqref{eq:policy_gradients} (cf., Alg.~\ref{algorithm:MARL_DRIVE}, l.~17). 
Sec.~\ref{subsec:problem_formulation} therefore describes the environment-level returns, while DRIVE defines how these are transformed into the shaped rewards that the RL updates optimize.
The non-negativity condition of advantage $\textit{TD}_{i}(\hat{u}_{t,i})$ in the DRIVE request is needed to expose defecting agents that typically have a greater advantage than the exploited agents in SDs \cite{axelrod1984evolution}.
If there is unilateral defective behavior, the defective agent $i$ will be penalized by the most exploited neighbor agent $j \in \mathcal{N}_{t,i}$, since $\Delta_{t,i,j} < 0$, which is ensured by the $\textit{min}$ aggregation terms\footnote{The $\textit{min}$ aggregations could be replaced by a sum or mean, similar to \cite{hughes2018inequity}. However, this would weaken the influence of the most exploited agent relative to others, thus tolerating individual dissatisfaction, which could spread in later epochs.} in Eq.~\eqref{eq:DRIVE_reward}.
However, if all agents act equally cooperative, then all $\Delta_{t,i,j} = 0$, and there is no additional reward or penalization.
The complete formulation of DRIVE at time step $t$ for any agent $i$ is given in Algorithm \ref{algorithm:DRIVE}. $\hat{V}_{i}$ is the approximated value function to calculate $\textit{TD}_{i}(\hat{u}_{t,i})$, according to Eq.~\eqref{eq:mi_td}, $\mathcal{N}_{t,i}$ is the current local neighborhood, $e_{t,i}$ is the experience tuple, and $\overline{u}_{i}$ is the current average reward of agent $i$.

\section{Related Work}

\textbf{MARL in Social Dilemmas. }
MARL has achieved substantial progress across a range of domains
\cite{tan1993multi,littman1994markov,bucsoniu2010multi,vinyals2019grandmaster}.
In decentralized SDs and SSDs, a key challenge is resolving misaligned
incentives without centralized control. Recent work addresses this through \emph{local}, peer-induced, or socially conditioned reward-shaping mechanisms \cite{hughes2018inequity,yang2020learning,phanJAAMAS2024,Altmann25-MEDIATE}, which operate in mixed-motive environments with observation-limited agents and thus align with our setting.
By contrast, much of the broader MARL literature --- methods for mitigating non-stationarity in cooperative settings \cite{bowling2002multiagent,matignon2007hysteretic,wei2016lenient}, globally informed reward shaping \cite{leibo2017multi,devlin2014potential}, opponent-shaping requiring access to opponents' parameters \cite{foerster2018learning,letcher2018stable,willi2022cola,zhao2022proximal}, and centralized incentive design \cite{yang2022adaptive,guresti2023iqflow} --- assumes global information or centralized coordination, and therefore does not extend to decentralized SDs.
\\[2pt]
\textbf{Peer Incentivization. }
Many PI methods have been introduced to promote mutual cooperation in a distributed fashion via reward exchange \cite{yi2022learning}. 
\emph{Gifting} extends the action space of each agent $i$ with a reward action to incentivize other agents $j \in \mathcal{N}_{t,i}$ \cite{lupu2020gifting}.
In contrast, DRIVE is built upon reciprocal reward shaping and does not require an extended action space.
\citet{schmid2021stochastic} proposes market-based PI where agents can establish bilateral agreements. A public sanctioning approach, where agents can reward or penalize each other based on known group behavior patterns, has been proposed in \cite{vinitsky2021learning}. \citet{hughes2018inequity} defines an \emph{inequity aversion (IA)} scheme based on non-negative reward differences weighted by two domain-dependent coefficients. 
\emph{Learning to Incentivize Other learning agents (LIO)} automatically learns an incentive function for each agent $i$ under full observability, based on the joint action of all other agents $j \neq i$ \cite{yang2020learning}. 
\citet{phanAAMAS22} propose Mutual Acknowledgment Token Exchange (MATE), a two‑phase protocol where agents exchange fixed tokens x: cooperating agents issue tokens while exploited agents receive them, yielding the modified payoff matrix shown in Fig. \ref{fig:payoff_tables_pi}.
Most PI methods are sensitive to the reward values of an environment due to relying on fixed incentive values or domain-dependent parameters, thus fail to cooperate when the reward values change (Fig. \ref{fig:intro_figure}).
Considering the PD as an example, Fig. \ref{fig:payoff_tables_pi} shows the modified payoff matrices of DRIVE, MATE, LIO, and IA. According to these matrices, DRIVE is the only method that does not depend on any particular hyperparameter or explicitly learned value. IA only modifies unilateral cooperation and defection but depends on two domain-dependent coefficients $\alpha$ and $\beta$ \cite{hughes2018inequity}. 
LIO learns an incentive function conditioning on the actions of the other agents, which requires sufficient experience and time to adapt accordingly. 
In case of MATE, the cooperation depends on an appropriate global token $\mathbf{x} \geq \textit{max}\{P-S, \frac{T-R}{3}\}$ that can be derived from the modified payoff matrix.
In the example instance in Fig. \ref{fig:PD_instance}, $\mathbf{x}$ should be at least 1, which is the original token value used for MATE in \cite{phanJAAMAS2024}. Interestingly, this value is used for other PI methods as well without further analysis or questioning \cite{lupu2020gifting,schmid2021stochastic}.
\\[2pt]
\textbf{Teaming. }
Beyond independent learners, recent work studies how social preferences, teams, or coalition structures shape incentives in mixed-motive settings, e.g., heterogeneous SVO-based reward shaping, team-based reward sharing, and adaptive reward mixing via price-of-anarchy minimization \cite{gemp2022d3c,radke2022exploring,mckee2020social,phan21vast}. 
These methods typically assume fixed or emergent group structures or modify each agent’s global utility function. In contrast, DRIVE introduces no explicit coalitions or shared rewards: incentives arise solely through local, bilateral reward exchanges, enabling decentralized alignment without predefined teams or population-level reward mixing.

\section{Theoretical Analysis}\label{sec:theoretical_analysis}
In the following, we analyze the incentive alignment properties of DRIVE in the general PD and its invariance to changing rewards. 
We emphasize that the mechanism does not guarantee convergence of learning dynamics, but rather reshapes incentives such that cooperation becomes the individually rational choice in repeated interactions. 
Although our formal discussion focuses on the canonical PD, many well-known sequential social dilemmas (SSDs), such as Coin and Harvest, instantiate PD-like incentive structures, as shown empirically in \cite{leibo2017multi} and explained in the Appendix. 
In Section~\ref{sec:results}, we further show that DRIVE also yields promising results in more complex SSDs with more than two agents.

\subsection{DRIVE in the General Prisoner's Dilemma}

In social dilemmas with a payoff table, as shown in Fig. \ref{fig:social_dilemma_matrix}, and inequalities $T > R > P > S$ of Eq.~\eqref{eq:PD_inequalities}, the DRIVE incentives do not change the payoffs $R$ for mutual cooperation and $P$ for mutual defection in the long run because the reward difference $\Delta_{t,i,j}$ would be zero for both agents.
With $R > P$, this will still favor mutual cooperation over mutual defection.
However, if there were (repeated) unilateral defection, the defective agent $i$ would send a request $T$ to the exploited agent $j$, which responds with a difference of $\Delta_{t,i,j} = S - T < 0$ in the worst case, where the reward of the defective agent $i$ would change to $T + \Delta_{t,i,j} = T + S - T = S$, while the reward of the exploited agent $j$ would change to $S - \Delta_{t,i,j} = S - (S - T) = T$, according to Eq.~\eqref{eq:DRIVE_reward}.
As shown in Fig. \ref{fig:payoff_tables_pi}, the DRIVE payoff switch of $T$ and $S$ enables both agents to overcome greed and fear, therefore incentivizing cooperation.

\begin{theorem}\label{theorem:drive_cooperation}
DRIVE aligns incentives in a generalized two-agent Prisoner’s Dilemma, making mutual cooperation a dominant strategy for both agents by reversing the temptation and sucker payoffs.
\end{theorem}

\begin{proof}
Using the reward difference $\Delta_{t,1,2} = \overline{u}_2 - \hat{u}_{t,1}$ to shape the rewards of both agents (Eq.~\ref{eq:DRIVE_reward}) results in a switch of the payoffs $T$ and $S$ in the original PD payoff matrix in Fig. \ref{fig:canonical_game}. According to the resulting payoff matrix in Fig. \ref{fig:payoff_tables_pi}, DRIVE agents are able to overcome greed, since $R > S$, as well as fear, since $T > P$. Thus, DRIVE agents are always incentivized to cooperate in the PD.
\end{proof}

Overall, DRIVE does not rely on fixed incentive values or domain-dependent parameters. Thus, it can cope with reward changes that still satisfy the PD inequalities w.r.t. greed and fear, as in the example of Fig. \ref{fig:intro_figure}.
For clarity, Appendix~\ref{app:detailed_PD} provides an extended example that concretely demonstrates the shaping mechanism in a single PD exchange step.
Note that Eq.~\eqref{eq:DRIVE_reward} effectively turns the stage game into a coordination game with $(C,C)$ as the unique Nash equilibrium. 
This reflects standard reward-design practice: instead of relying on learners to infer long-term externalities, we adjust instantaneous payoffs so that cooperation is individually rational while leaving the environment dynamics unchanged.
\begin{figure}[t]\centering
\includegraphics[width=0.4\textwidth]{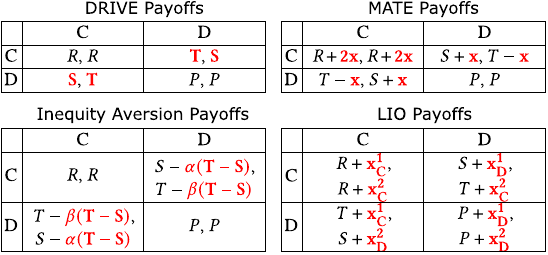}
\caption{Modified PD payoff matrices of different PI methods with payoff modifications highlighted in red \cite{hughes2018inequity,yang2020learning,phanJAAMAS2024}.}\label{fig:payoff_tables_pi}
\Description{Modified PD payoff matrices of different PI methods with payoff modifications highlighted in red \cite{hughes2018inequity,yang2020learning,phanJAAMAS2024}.}
\end{figure}

\subsection{Invariance to Changing Rewards}
We now analyze DRIVE and MATE w.r.t.\ their invariance to changing rewards in the PD, i.e., where the original payoffs $R$, $P$, $T$, and $S$ are dynamically altered by some external change function $f_{\textit{mod}}$.

\begin{definition}[Reward Change Function]\label{def:reward_change}
Let $r_{t,i}$ denote the original environmental reward of agent $i$ at time $t$.  
The modified reward in epoch $m$ is defined as $r'_{t,i} = f_{\textit{mod}}(r_{t,i}, m) = c_m r_{t,i} + b_m$, where $c_m > 0$ and $b_m \in \mathbb{R}$ are epoch-dependent scaling and shifting parameters that apply uniformly to all agents and timesteps within epoch $m$.  
This shared affine map preserves the strategic structure of the game (e.g., PD inequalities) while shifting reward magnitudes and offsets.
We focus on this broad but tractable class of per-epoch affine transformations: arbitrary schedules $c_m$ and $b_m$ cover all reward dynamics in Sec.~\ref{subsec:results_drift} and remain compatible with our per-epoch normalization scheme. 
Further illustrations and edge cases including epochs where $c_m \approx 0$ are provided in Appendix~\ref{app:affine_transformations}.
\end{definition}

\begin{theorem}\label{theorem:MATE_invariance}
MATE agents are not invariant to changing rewards in the general PD, where the environmental reward is altered by $f_{\textit{mod}}$.
\end{theorem}
\begin{proof}
According to the MATE payoff matrix in Fig. \ref{fig:payoff_tables_pi}, the global token $\mathbf{x} = x_1 = x_2 > 0$ must satisfy $\mathbf{x} \geq \textit{max}\{P-S, \frac{T-R}{3}\}$ for emergent cooperation. When the change function $f_{\textit{mod}}$ is chosen such that $f_{\textit{mod}}(P - S) > \mathbf{x}$ or $f_{\textit{mod}}(\frac{T-R}{3}) > \mathbf{x}$, then MATE agents are no longer guaranteed to cooperate mutually.
\end{proof}

\noindent Note that Theorem \ref{theorem:MATE_invariance} can be generalized to any other PI method that uses fixed peer incentive values.
One could in principle re‑tune $x$ after each change in $f_{\textit{mod}}$, but this requires additional global knowledge or meta‑optimization, whereas DRIVE adapts automatically through local reward differences.

\begin{theorem}\label{theorem:DRIVE_invariance}
DRIVE agents are invariant to changing rewards in the general PD, where the environmental reward is altered by $f_{\textit{mod}}$.
\end{theorem}

\begin{proof}
Altering the original payoffs $R$, $P$, $T$, and $S$ by any $f_{\textit{mod}}$ would proportionally change the DRIVE payoff matrix in Fig. \ref{fig:payoff_tables_pi}. Thus, DRIVE agents are always incentivized to cooperate mutually as the inequalities to overcome greed and fear remain satisfied.
\end{proof}

We now extend the invariance analysis from the PD to general SSDs using policy gradient methods with return normalization, as detailed in Appendix~\ref{subsec:appendix_normalization}.

\begin{theorem}\label{theorem:DRIVE_invariance_SSDs}
DRIVE agents are invariant to changing rewards in SSDs, where environmental rewards are altered by $f_{\textit{mod}}$, when trained with normalized policy gradient methods (Lemma~\ref{lemma:scaling_invariance} in Appendix~\ref{subsec:appendix_normalization}).
\end{theorem}

Consider sequential social dilemmas with returns normalized per epoch to zero mean and unit variance.  
If rewards are transformed by a positive affine map per epoch then both (i) standard policy gradient learning and (ii) DRIVE’s shaped rewards remain invariant with respect to the normalized return $\overline{G}_{t,i}$.  
Consequently, the incentive alignment induced by DRIVE persists under per-epoch reward scaling and shifting within this class of transformations.
Overall, under per-epoch normalization, both the base policy-gradient updates and DRIVE’s shaped rewards are invariant to shared per-epoch positive affine reward transformations (Lemma~\ref{lemma:scaling_invariance} in Appendix~\ref{subsec:appendix_normalization}).

\begin{proof}
The result follows from two components.  
First, Lemma~\ref{lemma:scaling_invariance} shows that standard policy gradient RL is invariant to positive affine reward transformations under per-epoch normalization.  
Second, DRIVE shapes rewards based on differences between epoch-average and instantaneous rewards. Applying $f_{\textit{mod}}$ gives $\Delta'_{t,i,j} = f_{\textit{mod}}(\overline{u}_i) - f_{\textit{mod}}(u_{t,i}) 
= f_{\textit{mod}}(\Delta_{t,i,j})$, so the shaped reward becomes $\hat{u}^{\textit{DRIVE}\,'}_{t,i} = f_{\textit{mod}}(u^{\textit{DRIVE}}_{t,i})$.
For $f_{\textit{mod}}(x) = c_m x + b_m$, normalization removes $c_m$ and $b_m$, leaving $\overline{G}_{t,i}$ unchanged.  
Hence, both the baseline policy gradient updates and DRIVE’s incentive effects are invariant to such transformations.
\end{proof}

\subsection{Effects of Non-Adherence to the Protocol}\label{subsec:non_adherence}

The preceding analysis assumes that all agents truthfully and synchronously follow the DRIVE protocol, i.e., each agent shares its epoch-average reward, responds to requests, and applies the shaping rule consistently. 
In practice, however, communication failures, delays, or strategic misreporting may occur. 
Here we analyze the impact of such \emph{non-adherence} compared to other PI approaches; for a detailed example, see Appendix~\ref{app:compliance}.
\\[2pt]
\textbf{Full compliance. }  
When all agents follow the protocol truthfully, the results of Theorem~\ref{theorem:drive_cooperation} and Theorems~\ref{theorem:DRIVE_invariance}--\ref{theorem:DRIVE_invariance_SSDs} apply.  
Payoffs $T$ and $S$ are swapped under unilateral defection, ensuring incentive alignment toward cooperation.
This reasoning can be generalized to SSDs beyond the PD, provided the characteristic inequalities for greed and fear are preserved.
\\[2pt]
\textbf{No requests sent (despite TD $\geq 0$). }  
An agent that withholds requests can still receive requests and respond truthfully.  
For its neighbors, this means one fewer reciprocal request. Due to the $\min$ aggregation, their penalization levels may increase slightly, but are not substantially affected as long as other agents send requests.  
For the non-requesting agent itself, $\hat u_{\textit{res}}=0$, so it forfeits the opportunity for reciprocal improvement and is indirectly penalized.  
If the agent defects, it may temporarily avoid penalization, but overall, it loses the benefits of mutual shaping and does not gain a stable advantage.  
\\[2pt]
\textbf{No responses sent. }  
If an agent does not respond to requests (due to communication loss or intentional withholding), the effect depends on the context:  
If the agent is defecting, withholding its negative $\Delta$ weakens the penalty on its exploited neighbor. However, it also loses the chance to benefit from reciprocity, so deliberate non-response is not rational.  
If the agent is cooperative, the missing (typically positive or small) $\Delta$ has limited effect because the $\min$ aggregation dampens single missing contributions.  
Thus, non-response is mostly neutral or slightly beneficial for defectors but costly overall.  
\\[2pt]
\textbf{False requests or responses. }  
Two types of misbehavior are possible:  
\emph{Sending requests with TD $<0$:} This can occur due to noise or early training misestimation. Such requests typically yield small or positive differences that are filtered out by the $\min$ operator and thus have little impact.  
\emph{Sending false responses:} An agent could deliberately misreport $\Delta$ to distort shaping. In the worst case, the $\min$ aggregation could impose unjustified penalties on requesters, potentially undermining cooperation. This depicts a fragility and limitation of DRIVE. 
\\[2pt]
\textbf{No compliance at all. }  
If all agents withhold or communication fails globally, then all shaping terms vanish. 
The system collapses to plain MARL without incentives, and mutual defection again becomes a possible equilibrium. 
This is similar to fixed-token methods like MATE, which also break down if their coordination mechanism is unavailable.

\subsection{Conceptual Discussion}\label{subsec:conceptual_discussion}
DRIVE offers a simple, adaptive PI mechanism for fostering cooperation in SDs with changing rewards. 
It operates through local peer interactions, requiring no central controller, global knowledge, or all-to-all communication. 
Similar to \cite{hughes2018inequity,yang2020learning,phanJAAMAS2024}, Algorithm~\ref{algorithm:DRIVE} scales linearly with $\mathcal{O}(4(n-1))$ in the worst case with respect to incentive exchanges.
At the same time, DRIVE relies on truthful peer communication to function correctly. 
Partial compliance leads to \emph{graceful degradation}: as long as at least one honest responder provides feedback, defectors are still penalized; if all agents fail to comply, DRIVE reduces to baseline MARL dynamics. 
Compared to fixed-token methods, DRIVE is at least as robust and often more adaptive, but its reliance on strict adherence to the exchange protocol introduces a new potential fragility. Designing protocol variants that are robust to partial compliance -- for example through redundant aggregation rules, stochastic auditing, or alternative shaping operators --- remains an important avenue for future work.
\\[2pt]
A key modeling assumption is the neighborhood $N_{t,i}$ that restricts which agents can exchange incentives. In our experiments, this neighborhood is tied to the environment as introduced in the subsequent section.
Crucially, DRIVE’s min-aggregation means a defector’s shaped reward depends only on its most exploited neighbor, so, unlike simple reward-sharing schemes, penalties do not dilute with larger populations. 
Appendix~\ref{app:compliance} extends the previous compliance analysis to general N-agent systems and arbitrary communication graphs, and formalizes the resulting robustness: defector penalization holds exactly when compliant agents form a dominating set.
Thus cooperation incentives degrade smoothly rather than collapsing abruptly under partial compliance.
\\[2pt]
As shown in the modified PD payoff tables in Fig. \ref{fig:payoff_tables_pi}, DRIVE only depends on the environmental rewards without requiring fixed incentive values $x$ as MATE, domain-dependent coefficients $\mathbf{\alpha}, \mathbf{\beta}$ as IA, or time-consuming incentive learning of $x^i_C, x^i_D$ as LIO. 
Note that given the characteristics of greed and fear exhibited by general PDs, the above theoretical analysis generally applies to various SDs.
Thus, DRIVE can be used in any SD, where rewards change dynamically through an external and unknown change function $f_{\textit{mod}}$. If the conditions for cooperation, i.e., characteristic inequalities of the SD, remain the same, DRIVE can adapt and maintain cooperation, as shown in Fig. \ref{fig:intro_figure} and later demonstrated.

\section{Experimental Setup\protect\footnotemark}\label{sec:experiments}
\footnotetext{Implementations are available at \url{https://github.com/philippaltmann/DRIVE}.}

We use three well-known SSDs based on \cite{foerster2018learning,perolat2017multi}. At every time step, the order of agent actions is randomized to resolve conflicts, e.g., when multiple agents step on a coin or tag each other simultaneously. All SSDs represent PD instances, as empirically shown in \cite{leibo2017multi} and Appendix~\ref{app:appendix_technical}. Since we modify the rewards by $f_{\textit{mod}}$, we assess the cooperation with domain-specific measures.
\\[2pt]
\textbf{Iterated Prisoner's Dilemma. } For IPD, we use the payoff matrix shown in Fig. \ref{fig:PD_instance}. Both agents observe the previous joint action $z_{t,i} = a_{t-1}$ at every time step $t$, which is the zero vector at the start state $s_{0}$. The Nash equilibrium is to always defect (DD). An episode consists of $H = 150$ iterations, and we set $\gamma=0.95$. The neighborhood $\mathcal{N}_{t,i} = \{j\}$ is defined by the other agent $j \neq i$. We measure the \emph{Cooperation Rate} as the cooperation (CC) count per episode divided by $H$.
\\[2pt]
\textbf{Coin. } \textit{Coin-2} and \textit{Coin-4} are SSDs and consist of $n \in \{2,  4\}$ agents with different colors, which start at random positions and have to collect a coin with a random color and position \cite{lerer2017maintaining,foerster2018learning}. If an agent collects a coin, it receives a reward of +1. However, if the coin has a different color than the collecting agent, another agent with the actual matching color is penalized with -2. After being collected, the coin respawns randomly with a new color.
All agents can observe the whole field and are able to move north, south, west, and east. Each agent can only determine if a coin has the same color as itself or not. An episode terminates after $H = 150$ time steps, and we set $\gamma=0.95$. The neighborhood $\mathcal{N}_{t,i} = \mathbb{D} - \{i\}$ is defined by all other agents $j \neq i$. We measure the \emph{``own coin" rate} $P(\textit{own coin}) = \frac{\textit{\# collected coins with same color}}{\textit{\# all collected coins}}$ based on the coins collected by each agent.
\\[2pt]
\textbf{Harvest. } \textit{Harvest-12} is an SSD and consists of $n = 12$ agents, starting at random positions and having to collect apples. The apple regrowth rate depends on the number of surrounding apples \cite{perolat2017multi}. 
If all apples are harvested, then no apple will grow anymore until the episode terminates.
At every time step, all agents receive a time penalty of -0.01. For each collected apple, an agent receives a reward of +1. All agents have a $7 \times 7$ field of view and are able to do nothing, move north, south, west, east, and tag other agents within their view with a tag beam of width 5 pointed to a specific cardinal direction. If an agent is tagged, it is unable to act for 25 time steps. Tagging does not yield any rewards. An episode terminates after $H = 250$ time steps, and we set $\gamma=0.99$.
The neighborhood $\mathcal{N}_{t,i}$ is defined by all other agents $j \neq i$ being in the field of view of agent $i$.
We measure the  \emph{Sustainability}, defined by the average number of time steps at which apples are collected.
\\[2pt]
\textbf{MARL Algorithms. }\label{subsec:marl_algorithms}
To isolate the impact of reward changes on PI mechanisms rather than the base RL algorithm, we implement all methods including DRIVE based on the standard policy gradient algorithm (Eq.~\ref{eq:policy_gradients}) with normalized returns (Appendix \ref{subsec:appendix_normalization}), as suggested in \cite{hessel2019multi}, which also serves as Naive Learning baseline without any PI-based reward shaping \cite{foerster2018learning}. 
We use LIO, MATE with $x = 1$ as suggested in \cite{phanJAAMAS2024}, and IA with $\alpha = 5$ and $\beta = 0.05$ as state-of-the-art PI baselines \cite{hughes2018inequity,yang2020learning,phanJAAMAS2024}.
For \textit{IPD} and \textit{Coin-2}, we directly include the performance of the opponent shaping techniques LOLA-PG and POLA-DiCE, as reported in \cite{foerster2018learning,zhao2022proximal}, due to the high computational demand of the second-order derivative calculation for deep neural networks.

\section{Experimental Results}\label{sec:results}

For each experiment, all respective algorithms were run 20 times over $E = 4,000$ epochs of 10 episodes to report the average progress and the 95\% confidence interval.

\subsection{Standard Setting -- Without Reward Changes}

\textbf{Setting. }
We compare DRIVE with the baselines in \textit{IPD}, \textit{Coin-2}, \textit{Coin-4}, and \textit{Harvest-12} without any reward change, i.e., $f_{\textit{mod}}(u,m) = u$ (Algorithm \ref{algorithm:MARL_DRIVE}, Line 15), in all epochs $m$.
\\[2pt]
\textbf{Results. }
The results are shown in Fig. \ref{fig:balance_results_star}. DRIVE achieves competitive and stable cooperation in \textit{IPD} compared with all baselines. In \textit{Coin-2}, DRIVE and MATE achieve the highest ``own coin" rate, significantly outperforming POLA-DiCE, which is more cooperative than LIO, IA, and LOLA-PG.
In both larger SSDs, namely \textit{Coin-4} and \textit{Harvest-12}, DRIVE and MATE achieve the highest level of cooperation, where DRIVE is slightly more cooperative in \textit{Coin-4}, while MATE is slightly more cooperative in \textit{Harvest-12}.

\begin{figure}[htb]\centering
\includegraphics[width=0.8\linewidth]{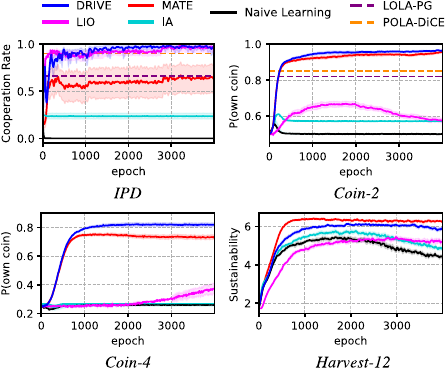}
\caption{Average progress of DRIVE and other baselines in SSDs without reward change. Shaded areas show the 95\% confidence interval. The results of LOLA-PG and POLA-DiCE are from \cite{foerster2018learning,zhao2022proximal}.}\label{fig:balance_results_star}
\Description{Average progress of DRIVE and other baselines in SSDs without reward change. Shaded areas show the 95\% confidence interval. The results of LOLA-PG and POLA-DiCE are from \cite{foerster2018learning,zhao2022proximal}.}
\end{figure}

\subsection{Dynamic Reward Changes}\label{subsec:results_drift}
\textbf{Setting. }
Next, we evaluate DRIVE and the baselines with different reward change functions $f_{\textit{mod}}$, as visualized in Fig. \ref{fig:drift_functions}:
\begin{enumerate}
\item \textbf{Linear increase}: $\> \> \> \> \> \> \> \> f^{I}_{\textit{mod}}(u,m) = u(\eta m + 1)$
\item \textbf{Exponential decay}: $ \> \>f^{II}_{\textit{mod}}(u,m) = ue^{-\eta m}$
\item \textbf{Stepwise increase}: $\> \> \> f^{III}_{\textit{mod}}(u,m) = u(\floor{\eta m} + \chi)$
\item \textbf{Damped cos.}: $f^{IV}_{\textit{mod}}(u,m) = \eta + u(1 - \frac{m}{E})\textit{cos}^2(2\eta m)$
\end{enumerate}
We set $\eta = 0.001$ and $\chi = 10$. We choose these functions to present different scenarios where the reward scale or shift either increases monotonically or converges to zero in the limit without altering the underlying inequalities of the SDs.
(1) could exemplify an external factor like increasing demands on the overall system. 
(2) could, e.g., simulate an internal factor like sensory wear or overall hardware decay over time. 
(3), on the other hand, demonstrates larger steps of change, as typically observed for shifts between simulations or simulation and reality.
Finally, (4) could be interpreted as the irregularities caused by human feedback or specifications being modified online.
\begin{figure}[t]\centering
\includegraphics[width=0.9\linewidth]{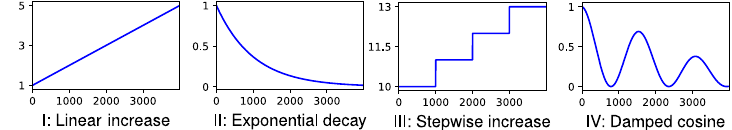}
\caption{The reward change functions $f_{\textit{mod}}$ used in each epoch $m$ according to Line 15 in Algorithm \ref{algorithm:MARL_DRIVE}.}\label{fig:drift_functions}
\Description{The reward change functions $f_{\textit{mod}}$ used in each epoch $m$ according to Line 15 in Algorithm \ref{algorithm:MARL_DRIVE}.}
\end{figure}
\begin{figure*}[t]\centering
\includegraphics[width=0.82\textwidth]{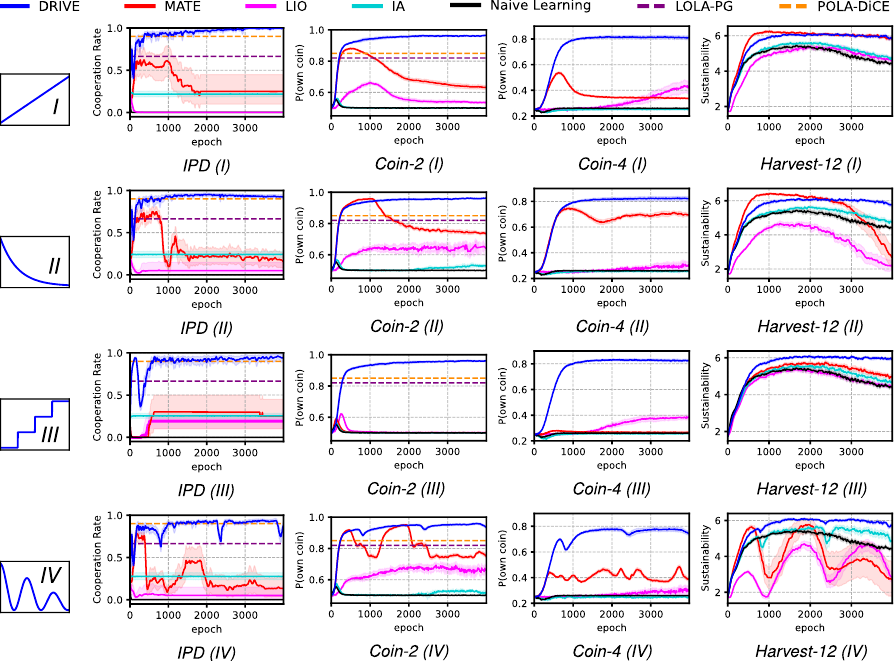}
\caption{Average progress of DRIVE and other baselines in SSDs with different reward change functions $f_{\textit{mod}}$ (I, II, III, and IV), as illustrated in the legend at the top and Fig. \ref{fig:drift_functions}. Shaded areas show the 95\% confidence interval. The results of LOLA-PG and POLA-DiCE are from \cite{foerster2018learning,zhao2022proximal}. The non-PI approaches Naive Learning, LOLA-PG, and POLA-DiCE are not affected by the reward changes of $f_{\textit{mod}}$ due to return normalization, as explained in Appendix \ref{subsec:appendix_normalization}.}\label{fig:balance_results_star_drift}
\Description{Average progress of DRIVE and other baselines in SSDs with different reward change functions $f_{\textit{mod}}$ (I, II, III, and IV), as illustrated in the legend at the top and Fig. \ref{fig:drift_functions}. Shaded areas show the 95\% confidence interval. The results of LOLA-PG and POLA-DiCE are from \cite{foerster2018learning,zhao2022proximal}. The non-PI approaches Naive Learning, LOLA-PG, and POLA-DiCE are not affected by the reward changes of $f_{\textit{mod}}$ due to return normalization, as explained in Appendix \ref{subsec:appendix_normalization}.}
\end{figure*}
\\[2pt]
\textbf{Results. }
The results are shown in Fig. \ref{fig:balance_results_star_drift}. The learning curve of DRIVE is not significantly affected by training with the functions $f^{I}_{\textit{mod}}$, $f^{II}_{\textit{mod}}$, and $f^{III}_{\textit{mod}}$. Function $f^{IV}_{\textit{mod}}$ causes occasional dips in the learning curve of DRIVE, which quickly recovers to its original level afterward.
The occasional dips are caused by the minima of the damped cosine function, where it is close to zero, which briefly nullifies the PD inequalities of Eq.~\eqref{eq:PD_inequalities}. However, this is a numeric issue rather than a conceptual one, as DRIVE quickly recovers afterward.
LIO is robust against any change function in \textit{Coin-4}, but its cooperation level significantly deteriorates in other domains. Despite automatically learning an incentive function, LIO is especially sensitive to $f^{IV}_{\textit{mod}}$. 
IA is affected in \textit{Coin-2} and \textit{Coin-4}, where it never outperforms Naive Learning. 
MATE is only able to resist reward changes in \textit{Harvest-12} when training with $f^{I}_{\textit{mod}}$.
In all other cases, the cooperation level of MATE significantly deteriorates, which is even worse than Naive Learning in some cases. The learned behavior is least stable with $f^{IV}$, where MATE is not able to recover from dips in the learning curve, unlike DRIVE.
DRIVE is the only PI method that outperforms Naive Learning, LOLA-PG, and POLA-DiCE in \textit{IPD} and \textit{Coin-2}.
Appendix~\ref{sec:additional_results} provides further cooperation metrics (social welfare, equality, sustainability, and peace) in \textit{Harvest-12}.

\section{Discussion}

We presented DRIVE, an adaptive PI approach to achieve and maintain cooperation in SDs with changing rewards. DRIVE agents reciprocally exchange reward differences to incentivize mutual cooperation in a fully decentralized manner.
\\[2pt]
Our theoretical analysis shows how DRIVE achieves incentive alignment toward mutual cooperation in the general Prisoner's Dilemma by overcoming greed and fear while remaining invariant to per-epoch reward shifts and scaling. This contrasts with prior PI methods that rely on fixed incentive values, domain-dependent coefficients, or time-consuming incentive learning. Consequently, DRIVE is better aligned with game-theoretic assumptions and safer to deploy in environments with drifting rewards.
\\[2pt]
Our experiments further demonstrate that DRIVE significantly outperforms state-of-the-art methods such as LIO, MATE, and IA, as well as opponent-shaping baselines like POLA-DiCE and LOLA-PG, across multiple SSDs with dynamic reward changes. The results confirm the sensitivity of existing PI approaches to reward magnitudes and highlight DRIVE’s ability to sustain cooperation under such transformations, provided the underlying strategic inequalities remain intact.
\\[2pt]
Despite these promising results, DRIVE relies on truthful and synchronous peer communication. Partial non-compliance leads to graceful degradation in performance. Moreover, the current theoretical analysis assumes homogeneous populations and complete or well-connected communication graphs. These assumptions may not hold in more realistic settings with strategic manipulation, communication noise, or sparse interaction structures.
\\[2pt]
Several directions for future work follow from these observations. One avenue is to develop mechanisms that are robust to partial or strategic non-compliance, for example, through redundant aggregation, stochastic auditing, or consensus systems \cite{Altmann25-MEDIATE}.
Extending the theoretical analysis to heterogeneous populations and more general network topologies represents another important step toward real-world applicability. 
A natural next step is to study learned and dynamic neighborhoods, and to combine DRIVE with models of evolving social structure and coalitions, as suggested by recent work on reward sharing in teams and coalitions.
Finally, incorporating agent identification could enable targeted bilateral responses, e.g., against adversarial peers.

%%% The next two lines define, first, the bibliography style to be 
%%% applied, and, second, the bibliography file to be used.

\bibliographystyle{ACM-Reference-Format} 
\bibliography{DRIVE}\balance

\clearpage\onecolumn\appendix
\section{Extended Theoretical Results}\label{app:extended_theory}

In this section, we provide additional worked-out examples and detailed arguments to complement the theoretical analysis in the main body. 
These results give a more concrete and intuitive understanding of how DRIVE operates in simple settings, how its behavior changes under different assumptions of protocol adherence, and how the microscopic request/response mechanism underpins the macroscopic equilibrium results, and provides a simple extension to larger multi-agent settings. 
\\[4pt]
\textbf{Preliminaries. }
We recall the standard two-player Prisoner’s Dilemma (PD) with actions $C$ (cooperate) and $D$ (defect) and payoffs $(R,R)$, $(P,P)$, $(T,S)$, and $(S,T)$ for profile $(C,C)$, $(D,D)$, $(D,C)$, and $(C,D)$, respectively (Fig.~\ref{fig:social_dilemma_matrix}). The game is a PD if $T > R > P > S$ (cf. Eq.~\ref{eq:PD_inequalities}), where $T>R$ encodes greed and $P>S$ encodes fear~\cite{axelrod1984evolution}.
We denote the environmental reward of agent $i$ at time $t$ by $u_{t,i}$, the externally modified reward in epoch $m$ by $\hat{u}_{t,i}=f_{\textit{mod}}(u_{t,i}, m)$, and the running average of $\hat{u}_{t,i}$ within epoch $m$ by $\bar{u}_i$.
The DRIVE shaping rule for agent $i$ at time $t$ is
\begin{equation}
u^{\mathrm{DRIVE}}_{t,i} = \hat{u}_{t,i} - \min_{j \in N_{t,i}}\{\Delta_{t,j,i}\} + \min_{j \in N_{t,i}}\{\Delta_{t,i,j}\},
\label{eq:DRIVE_reward_app}
\end{equation}
where $\Delta_{t,i,j} = \bar{u}_j - \hat{u}_{t,i}$ is the response sent from $j$ to $i$ when $i$ issues a request and the temporal difference gate TD$_i(\hat{u}_{t,i})$ in Eq.~\eqref{eq:mi_td} is non-negative. 
In the 2-agent PD we have $N_{t,1}=\{2\}$ and $N_{t,2}=\{1\}$ so the minima reduce to single terms.

\subsection{Single Exchange Mechanism Underlying Incentive Alignment in a 2-Player Prisoner’s Dilemma}\label{app:detailed_PD}

We give a concrete, single-step derivation showing how the DRIVE request/response protocol (Alg.~\ref{algorithm:DRIVE}) together with the TD gate in Eq.~(3) and the shaping rule in Eq.~(4) yields the \emph{payoff swap} $T \leftrightarrow S$ under unilateral defection.  
This turns the PD into a coordination game with $(C,C)$ as the unique Nash equilibrium.
Unless stated otherwise, we assume full protocol compliance.
\\[4pt]
\textbf{Setup. } 
Consider a single time step $t$ in a 2-agent PD with payoffs
$T>R>P>S$ and neighborhoods $N_{t,1}=\{2\}$,
$N_{t,2}=\{1\}$. Rewards may already be transformed by
$f_{\textit{mod}}$ in epoch $m$; for notational convenience we
write
$$
\hat{T} := f_{\textit{mod}}(T,m),\quad
\hat{R} := f_{\textit{mod}}(R,m),\quad
\hat{P} := f_{\textit{mod}}(P,m),\quad
\hat{S} := f_{\textit{mod}}(S,m).
$$
We assume that, in the long run under stationary play, epoch averages coincide with the corresponding instantaneous payoffs for each action profile (e.g., under $(C,C)$ both agents have $\bar{u}_i = \hat{R}$).
\\[4pt]
\textbf{Mutual cooperation and mutual defection. } 
If both agents cooperate at time $t$, each receives $\hat{u}_{t,1}=\hat{u}_{t,2}=\hat{R}$ and, in steady state, $\bar{u}_1=\bar{u}_2=\hat{R}$. 
Hence $\Delta_{t,1,2} = \bar{u}_2 - \hat{u}_{t,1} = 0$ and $\Delta_{t,2,1}=0$, so Eq.~\eqref{eq:DRIVE_reward_app} yields
$$
u^{\mathrm{DRIVE}}_{t,1}=\hat{R},\quad
u^{\mathrm{DRIVE}}_{t,2}=\hat{R}.
$$
A similar argument shows that under mutual defection both agents retain payoff $\hat{P}$. 
Thus DRIVE leaves $R$ and $P$ unchanged in the long run.
\\[4pt]
\textbf{Unilateral defection. }
Now consider a state where agent~1 defects ($D$) and agent~2 cooperates ($C$), so the instantaneous modified rewards are $\hat{u}_{t,1} = \hat{T}$, $\hat{u}_{t,2} = \hat{S}$. 
Under the usual PD ordering, the defector benefits from a higher reward than the cooperator. 
We assume that this yields a non-negative TD advantage for the defector and a non-positive one for the cooperator, so that
\begin{equation}
\mathrm{TD}_1(\hat u_{t,1}) \ge 0 \;\Rightarrow\; \text{agent 1 sends a request},\qquad
\mathrm{TD}_2(\hat u_{t,2}) < 0 \;\Rightarrow\; \text{agent 2 does not send}.
\end{equation}
This corresponds to the intended operating regime of DRIVE and matches the empirical behavior observed in our experiments.
Agent~1 sends its instantaneous reward $\hat u_{t,1}=\hat T$ to agent~2 as a request. Upon receiving this, agent~2 computes the (epoch-$m$) difference $\Delta_{t,1,2} \;=\; \overline{u}_2 - \hat u_{t,1}$.
In the PD step considered, the worst-case consistent choice is $\overline{u}_2 = \hat S$, yielding $\Delta_{t,1,2} \;=\; \hat S - \hat T \;<\; 0$.
Agent~2 returns this value to agent~1. Because only agent~1 sent a request, the reverse difference $\Delta_{t,2,1}$ is undefined in this step; by Algorithm~\ref{algorithm:DRIVE} it defaults to $0$ for aggregation when no responses arrive.
Using the responder’s epoch-average $\overline{u}_j$ rather than its instantaneous reward makes DRIVE sensitive to systematic exploitation instead of single noisy outcomes. Only agents whose recent average return is consistently lower than a neighbor’s request generate strong negative $\Delta$ terms. 
If both sides used instantaneous rewards, a single lucky or unlucky step could trigger large transfers, making incentives noisy and easy to game via short-term risk-taking. Using instantaneous rewards against averages would instead turn DRIVE into a risk-sharing mechanism, while comparing averages on both sides would be very stable but slow to react. The chosen average–instant design therefore balances robustness to noise with responsiveness to emerging exploitation.
In the 2-player case, the minima over neighbors reduce to single terms. The per-step shaped rewards are
\begin{align}
\hat u^{\mathrm{DRIVE}}_{t,1} 
&= \hat u_{t,1} \;-\; \underbrace{\min\{\Delta_{t,2,1}\}}_{=\,0} \;+\; \underbrace{\min\{\Delta_{t,1,2}\}}_{=\,\hat S-\hat T}
\;=\; \hat T + (\hat S-\hat T) \;=\; \hat S, \label{eq:defector_swaps_to_S}\\[4pt]
\hat u^{\mathrm{DRIVE}}_{t,2} 
&= \hat u_{t,2} \;-\; \underbrace{\min\{\Delta_{t,1,2}\}}_{=\,\hat S-\hat T} \;+\; \underbrace{\min\{\Delta_{t,2,1}\}}_{=\,0}
\;=\; \hat S - (\hat S-\hat T) \;=\; \hat T. \label{eq:cooperator_swaps_to_T}
\end{align}
Thus, the unilateral defection step $(\hat T,\hat S)$ is reshaped to $(\hat S,\hat T)$: the defector receives the sucker payoff and the cooperator the temptation payoff. 
Eqs.~(\ref{eq:defector_swaps_to_S}--\ref{eq:cooperator_swaps_to_T}) explicitly demonstrate how the TD gate together with the DRIVE shaping rule realizes the $T\leftrightarrow S$ swap in a single exchange step.
By symmetry, the profile $(C,D)$ is reshaped from $(\hat{S},\hat{T})$ to $(\hat{T},\hat{S})$.
\\[4pt]
\textbf{Resulting incentives and equilibrium interpretation. }
Collecting all action profiles, DRIVE induces the modified payoff matrix
% $$\begin{array}{c|cc}& $C$ & D \\\hline $C$ & (\hat{R},\hat{R}) & (\hat{T},\hat{S}) \\ $D$ & (\hat{S},\hat{T}) & (\hat{P},\hat{P}) \end{array}$$
\begin{table}[ht]\begin{tabular}[b]{|P{3em}|P{3em}|P{3em}|}
\hline & $C$ & $D$ \\ \hline 
$C$ & $\hat R$, $\hat R$ & $\hat S$, $\hat T$ \\ \hline 
$D$ & $\hat T$, $\hat S$ & $\hat P$, $\hat P$ \\ \hline
\end{tabular}\end{table}
with $\hat{T}>\hat{R}>\hat{P}>\hat{S}$.
Hence, if the opponent plays $C$, responding with $C$ yields $\hat R>\hat S$, and if the opponent plays $D$, responding with $C$ yields $\hat T>\hat P$. Cooperation is therefore a strict best response to both actions, so $(C,C)$ is the unique Nash equilibrium.
\\[4pt]
\textbf{Takeaway. }
Under the TD gate (Eq.~\ref{eq:mi_td}) and shaping rule (Eq.~\ref{eq:DRIVE_reward}), a single request/response exchange in a unilateral defection state deterministically maps $(\hat T,\hat S)\mapsto(\hat S,\hat T)$, eliminating greed and fear incentives in the transformed PD for that step. 
This microscopic mechanism directly underpins the equilibrium-level result that DRIVE turns the Prisoner’s Dilemma into a coordination game with cooperation as the unique Nash equilibrium, providing a worked-out complement to the proof of Theorem~\ref{theorem:drive_cooperation}.

\subsection{Analysis of Protocol Compliance Cases}\label{app:compliance}

To complement the discussion in Section~\ref{subsec:non_adherence}, we present a worked example for a two-player Prisoner’s Dilemma under DRIVE with different levels of protocol compliance.
We consider generic PD payoffs $T>R>P>S$ as in Eq.~(1); for concreteness one may think of the canonical values $(T,R,P,S)=(5,3,1,0)$, but none of the arguments below depend on these particular numbers.
In each case, agent $i$ defects while agent $j$ cooperates, and DRIVE applies reward shaping based on request/response behavior.
\begin{table}[ht]\centering
\begin{tabular}{p{0.32\linewidth}p{0.28\linewidth}p{0.28\linewidth}}
\toprule
\textbf{Case} & \textbf{Defector payoff} & \textbf{Cooperator payoff} \\
\midrule
(1) Full compliance & $S$ & $T$ \\
(2) Defector does not send request & $T$ (no penalty) & $S$ \\
(3) Cooperator withholds response & $T$ (no penalty) & $S$ \\
(4) Misreporting (false request/response) & unstable; may collapse penalties & unstable; may lead to mis-penalization \\
(5) No compliance at all & $T$ & $S$ \\
\bottomrule
\end{tabular}
\caption{Illustration of payoffs under different protocol compliance scenarios in a 2-player PD.} \label{tab:compliance_cases}
\vspace{-1em}
\end{table}
% \\[4pt]
\textbf{(1) Full compliance. }
If all agents truthfully follow the protocol, the results of Theorem~\ref{theorem:drive_cooperation} apply.
Unilateral defection $(T,S)$ is reshaped into $(S,T)$: the defector receives $S$ and the cooperator $T$, eliminating the incentive to defect.
\\[4pt]
\textbf{(2) Defector does not send a request. }
If the defector refrains from sending a request despite positive TD, it avoids penalization from neighbors but also loses the chance of mutual improvement.
Its payoff remains $T$, while the cooperator remains at $S$.
This strategy indirectly penalizes the non-requesting agent, which cannot gain from reciprocal shaping.
\\[4pt]
\textbf{(3) Cooperator withholds a response. }
If the cooperator fails to respond (intentionally or due to loss), the defector avoids the intended penalty.
The payoffs revert to $(T,S)$.
Because of the min-aggregation across neighbors, the absence of a single response only matters if it is the unique or strongest penalty.
\\[4pt]
\textbf{(4) Misreporting. }
Misreporting can occur in two forms: (i) agents send requests despite having no improvement signal (TD$<0$), or (ii) agents respond with incorrect values of $\Delta$ (due to error or intent).
In the first case, the resulting differences are typically small or positive, which the min operator filters out, making the effect negligible.
In the second case, however, false responses can distort the aggregation:
a misreporting neighbor may reduce or nullify the penalty for a defector, or unjustly penalize a compliant requester.
While this can destabilize cooperation in the worst case, it is not worse than the vulnerabilities faced by other PI methods.
Overall, misreporting represents a robustness limitation but does not fundamentally break the mechanism as long as the majority of responses remain truthful.
\\[4pt]
\textbf{(5) No compliance at all. }
If no requests or responses are exchanged, all shaping terms vanish. DRIVE collapses to plain MARL with payoffs $(T,S)$ under unilateral defection, and cooperation incentives are lost.
\\[4pt]
\textbf{Extension to larger scenarios. }
The two-player example above presents the limiting case where any single non-compliant agent corresponds to $50\%$ protocol adherence, which strongly magnifies the impact of individual deviations.
In larger populations, however, the effect of partial non-compliance is more gradual: as long as a sufficient fraction of agents continue to respond truthfully, the min-aggregation still enforces penalization of defectors.
\\[4pt]
\textbf{Three-player example. }
Consider a group of three agents, where one defects ($i$) and two cooperate ($j,k$).
Each cooperator truthfully responds with $\Delta_{i,j}, \Delta_{i,k} < 0$.
The defector’s shaped reward is $$u^{\text{DRIVE}}_i = T + \min\{\Delta_{i,j}, \Delta_{i,k}\},$$ so as long as at least one cooperator responds truthfully, $u^{\text{DRIVE}}_i$ is penalized.
If one cooperator withholds, the other still enforces the penalty, though possibly weaker.
Thus, unlike in the 2-player case, cooperation incentives do not collapse entirely when a single agent is non-compliant.

\subsection{Extension to a Simple N-Agent PD Setting}\label{app:multi_PD}

The main theorem focuses on the 2-agent Prisoner’s Dilemma (PD).  
Here we first provide a clean extension to an $N$-agent \emph{graphical} PD under full protocol compliance, where payoffs decompose into symmetric 2-player PD interactions, and then progressively relax structural and behavioral assumptions.
\begin{definition}[Graphical PD]
Let $G=(V,E)$ be an undirected graph with $|V|=N$ agents.
Each agent $i\in V$ chooses an action $a_i\in\{C,D\}$. 
For every edge $(i,j)\in E$ we associate a 2-player PD with payoffs $(R,R)$, $(P,P)$, $(T,S)$, $(S,T)$ as in Eq.~\eqref{eq:PD_inequalities}. The stage payoff of agent $i$ is
$$u_i(a) = \sum_{j\in N(i)} u_{i,j}(a_i,a_j),$$
where $N(i)$ is the neighborhood of $i$ in $G$ and $u_{i,j}$ is the PD payoff against neighbor $j$.
\end{definition}

To isolate the effect of incentive alignment, we assume that each pairwise interaction is shaped independently using the 2-agent DRIVE protocol, and that all agents comply fully with the request and response rules. 
That is, in every interaction with neighbor $j$, agent $i$ behaves exactly as in the 2-agent analysis of Appendix~\ref{app:detailed_PD}.
\begin{proposition}[Pairwise incentive alignment]\label{prop:multi_PD}
Consider a graphical PD as above and assume that each edge $(i,j)$ uses the 2-agent DRIVE protocol with shaping rule \eqref{eq:DRIVE_reward_app} applied independently to $u_{i,j}(a_i,a_j)$. 
Then, for every agent $i$ and every joint action profile $a_{-i}$ of its neighbors,
$$
u_i^{\mathrm{DRIVE}}(C,a_{-i}) > u_i^{\mathrm{DRIVE}}(D,a_{-i}),
$$
so cooperation is a dominant action for every agent and the unique Nash equilibrium of the shaped game is the all-cooperate profile $(C,\dots,C)$.
\end{proposition}
\begin{proof}
By Theorem~\ref{theorem:drive_cooperation} and its detailed proof in Appendix~\ref{app:detailed_PD}, for each neighbor $j$ the pairwise DRIVE-shaped payoff function $u_{i,j}^{\mathrm{DRIVE}}$ makes $C$ strictly dominate $D$ in the corresponding 2-player PD. That is, for any fixed $a_j\in\{C,D\}$, $u_{i,j}^{\mathrm{DRIVE}}(C,a_j) > u_{i,j}^{\mathrm{DRIVE}}(D,a_j)$.
Summing these strict inequalities over all neighbors $j\in N(i)$ for any joint neighbor action profile $a_{-i}$ yields 
$$
u_i^{\mathrm{DRIVE}}(C,a_{-i}) = \sum_{j\in N(i)} u_{i,j}^{\mathrm{DRIVE}}(C,a_j) > \sum_{j\in N(i)} u_{i,j}^{\mathrm{DRIVE}}(D,a_j) = u_i^{\mathrm{DRIVE}}(D,a_{-i}).
$$
Hence $C$ is a strict best response to any $a_{-i}$ for all agents $i$, so $(C,\dots,C)$ is the unique Nash equilibrium.
\end{proof}

This proposition formalizes how the 2-player incentive alignment result extends to a simple class of $N$-agent games with additive pairwise PD interactions. 
Our SSD experiments instantiate more complex interaction patterns, but the graphical PD result provides a clean theoretical illustration of how DRIVE scales beyond two agents.
\\[4pt]
\textbf{Pairwise PD interactions across neighborhoods.}
In many multi-agent systems, agents repeatedly engage in pairwise social dilemmas with dynamically changing or randomly matched partners (an \emph{interaction graph}), while incentive exchanges are restricted to a possibly different and fixed \emph{communication
graph} defined by neighborhoods $N_{t,i}$.
In such settings, the DRIVE-shaped reward of a defector $i$ still takes the form $u^{\mathrm{DRIVE}}_i = T + \min_{j \in N_{t,i}} \Delta_{i,j}$.
Whenever the exploited cooperator in a given PD round is also a neighbor of $i$, the exact $T \leftrightarrow S$ payoff swap from the 2-agent case is recovered. 
When interaction partners lie
outside $N_{t,i}$, the penalty is instead delivered by whichever compliant neighbors suffer a sustained long-run disadvantage from $i$’s behavior (e.g., through shared resources or future interactions).
Thus, the interaction graph determines \emph{which behaviors create negative externalities}, while the communication graph determines \emph{who can hold whom accountable}.
As long as each defector has at least one compliant neighbor whose long-run return is reduced when it defects, the sign of the incentive remains cooperative, even if the magnitude differs from the fully local 2-agent PD.
\\[4pt]
\textbf{Beyond graphical PDs.}
Many benchmark social dilemma environments, such as Coin Game or Harvest, do not admit an explicit decomposition into pairwise PD payoffs. 
Instead, incentives arise indirectly through shared resources, delayed consequences, and aggregate system dynamics.
Nevertheless, the same principle applies: agents whose behavior systematically reduces their neighbors’ long-run returns generate negative differences $\Delta_{i,j}$ and are penalized by DRIVE.
This perspective explains why DRIVE empirically  romotes cooperation in such environments, even though no explicit PD structure is present.
\\[4pt]
\textbf{General $N$-agent compliance and robustness.}
The preceding results assume full protocol compliance. 
We now consider a general $N$-agent setting in which only $M$ agents truthfully follow the DRIVE protocol, while $K=N-M$ agents may behave adversarially by withholding requests or responses or by misreporting.
For any game-theoretic defector $i$, effective penalization requires two conditions: 
(i) $i$ itself must comply with the request gate
(sending a request when $\mathrm{TD}\ge 0$), and 
(ii) at least one neighbor $j\in N_i$ must respond truthfully. 
Under these conditions, the shaped reward of $i$ is
$$
u^{\mathrm{DRIVE}}_i = T + \min\{\Delta_{i,j}\}_{j\in N_i,\; j\ \mathrm{compliant}}.
$$
As long as a compliant responder exists, defection is penalized; 
the severity depends on the most negative difference among the responders, but the penalty does not collapse entirely.
In fully connected populations, this reduces to the simple requirement that at least two agents comply (the requester and one responder), i.e.\ $M/N \geq 2/N$.
In general communication topologies, the compliant set must form a \emph{dominating set}, ensuring that each requester has at least one compliant neighbor.
Thus, DRIVE exhibits a \emph{graceful degradation}: as compliance $M/N$ decreases, cooperation incentives weaken smoothly but only collapse completely if no responder is available.
\\[4pt]
\textbf{Connectivity requirements.}
In fully connected graphs, penalization is ensured as soon as at least two agents comply (the requester and one responder), corresponding to a compliance rate $c \ge 2/N$. 
In general communication topologies, the compliant agents must form a \emph{dominating set}: every requester must have at least one compliant neighbor. 
Formally, for graph $G=(V,E)$ with compliant set $C\subseteq V$, penalization holds iff $\forall i\in V:\; N(i)\cap C\neq\varnothing$. 
The minimal compliant set size is the domination number $\gamma(G)$, yielding a topology-aware threshold $|C|\ge \gamma(G)$.
\\[4pt]
\textbf{Adversarial robustness.}
The above threshold can be interpreted directly in terms of $K$ adversarial agents.
If $K$ agents are adversarial (never providing truthful responses), DRIVE continues to penalize every \emph{protocol-compliant} defector as long as the remaining $M=N-K$ compliant agents contain a dominating set. 
Equivalently, the system tolerates up to $K\le N-\gamma(G)$ adversarial agents.
However, agents that are adversarial both in their \emph{actions} (always defecting) and their \emph{protocol behavior} (never sending requests) cannot be punished by any PI scheme based on self-reported signals, including DRIVE.
To compliant neighbors these agents simply appear as players whose defections generate no negative differences; this reduces local cooperative pressure and gradually pushes the system toward the underlying non-PI MARL dynamics (Case~5), but does not cause an abrupt collapse.
\\[4pt]
Overall, these results highlight DRIVE’s property of \emph{graceful degradation}: full compliance recovers strong incentive-alignment guarantees, while partial compliance weakens but does not immediately destroy cooperative incentives, provided the communication topology satisfies minimal coverage conditions.
In two-agent systems, a single non-compliant agent already reduces compliance to $50\%$, eliminating penalization.
By contrast, in larger populations ($N>2$), even if several agents defect from the protocol, the remaining majority can still enforce penalization through the min-aggregation.
Hence, DRIVE remains effective as long as the compliance rate stays above the topology-specific threshold, often close to a simple majority in practice.
Nevertheless, systematic misreporting or universal non-compliance collapses DRIVE to baseline MARL, emphasizing the importance of designing variants that are robust to missing or adversarial responses in future work.

\subsection{Affine Reward Transformations: Scope and Edge Cases}\label{app:affine_transformations}

Here, we expand on the affine reward-change model used in our analysis. In each epoch $m$, the environmental reward is transformed by a shared positive affine map $r' = c_m\, r + b_m$, applied uniformly to all agents and timesteps. The schedules $c_m$ and $b_m$ may vary arbitrarily across epochs, covering all reward dynamics used in Sec.~\ref{subsec:results_drift} (linear, exponential, stepwise, and damped–cosine).
\\[4pt]
\textbf{Example: damped–cosine schedule. }
For illustration, in the damped–cosine setting we use
$c_m = (1 - m/E)\cos^2(2\eta m)$ and $b_m = \eta$, so the effective reward scale oscillates while decaying over time.
\\[4pt]
\textbf{Degenerate epochs where $c_m \approx 0$. }
Isolated epochs may satisfy $c_m \approx 0$, yielding $r' \approx b_m$ and temporarily collapsing the Prisoner’s Dilemma inequalities ($T, R, P, S$ become indistinguishable). The stage game becomes nearly payoff-indifferent in these epochs. Our results treat such cases explicitly by applying the PD analysis only to epochs with $c_m > 0$.
\\[4pt]
\textbf{Transformations not covered. }
This affine class excludes genuinely non-linear reward changes such as clipping, saturation, sign-dependent remapping, or state-dependent transformations that cannot be written as a shared affine map. These represent a broader family of sim-to-real shifts that we leave to future work.

\subsection{Reward Changes and Invariance in the PD}\label{app:invariance_PD}

We first formalize how affine reward transformations affect the
PD inequalities.

\begin{lemma}[Affine transformations preserve PD structure]\label{lem:affine_PD}
Let $T>R>P>S$ and define $T'=cT+b$, $R'=cR+b$, $P'=cP+b$, $S'=cS+b$ for some $c>0$, $b\in\mathbb{R}$. Then $T'>R'>P'>S'$.
\end{lemma}

\begin{proof}
Since $c>0$, we have $T'-R' = c(T-R) > 0$, $R'-P' = c(R-P) > 0$, and $P'-S' = c(P-S) > 0$, so $T'>R'>P'>S'$. The additive constant $b$ cancels in all pairwise differences.
\end{proof}

We can now restate Theorems~\ref{theorem:MATE_invariance} and \ref{theorem:DRIVE_invariance} more formally.

\begin{proof}[Proof of \autoref{theorem:MATE_invariance}]
In MATE, mutual cooperation is individually rational iff the global token $x>0$ satisfies $x \ge \max\{P-S,(T-R)/3\}$~\cite{phanJAAMAS2024}. 
Under the affine transformation of Definition~\ref{def:reward_change}, these thresholds become $P'-S'$ and $(T'-R')/3$, i.e., 
$$
x \ge \max\{P'-S', (T'-R')/3\}.
$$
If we fix any $x$ satisfying this condition, because $c_m$ and $b_m$ are arbitrary (subject to $c_m>0$), we can always choose $(c_m,b_m)$ such that $P'-S' > x$, e.g., by taking $c_m$ large enough.
Then the transformed threshold exceeds $x$ and the MATE payoff matrix no longer guarantees cooperation as a best response. 
Hence MATE is not invariant to such reward changes.
\end{proof}
\begin{proof}[Proof of \autoref{theorem:DRIVE_invariance}.]
Let $(R,P,T,S)$ and $(R',P',T',S')$ be two sets of PD payoffs related by any $f_{\textit{mod}}$ of Definition~\ref{def:reward_change}. 
By Lemma~\ref{lem:affine_PD}, $T>R>P>S$ implies $T'>R'>P'>S'$. 
The detailed analysis in Appendix~\ref{app:detailed_PD} shows that, under full compliance, DRIVE leaves the mutual-cooperation and mutual-defection payoffs unchanged and swaps the temptation and sucker payoffs under unilateral defection. 
Applying the same calculation to $(R',P',T',S')$ shows that the transformed payoffs obey the same pattern. 
In both cases the resulting game has $(C,C)$ as its unique Nash equilibrium, so the dominance of
cooperation is preserved. 
\end{proof}

\subsection{Training with Normalized Returns}\label{subsec:appendix_normalization}

In our paper, $\hat{\pi}_{i}$ and $\hat{V}_{i}$ are trained with \emph{normalized returns} $\overline{G}_{t,i} = \frac{1}{\sigma_{i}}(G_{t,i} - \mu_{i})$, where $\mu_{i}$ and $\sigma_{i}$ are the mean and standard deviation of all returns $G_{t,i}$ of agent $i$ in an epoch $m$, to improve training stability \cite{van2016learning}. The normalization makes standard policy gradient RL (without PI mechanism) invariant to reward changes, where rewards are scaled by a factor $c_m > 0$ or shifted by a scalar $b_m \in \mathbb{R}$, according to Lemma~\ref{lemma:scaling_invariance}.

\begin{lemma}\label{lemma:scaling_invariance}
Standard policy gradient RL (without PI mechanism), according to Eq. \ref{eq:policy_gradients}, is invariant to reward changes, where the original environmental reward is scaled by a factor $c_m > 0$ or shifted by a scalar $b_m \in \mathbb{R}$, when the obtained returns $G_{t,i}$ in an epoch $m$ are normalized to zero mean and standard deviation such that $\overline{G}_{t,i} = \frac{1}{\sigma_{i}}(G_{t,i} - \mu_{i})$.
\end{lemma}

\begin{proof}
Since $c_m > 0$, the normalized return $\overline{G}_{t,i}$ is unaffected by the scaling and the shift with any $B_m = \sum^{H-1}_{t=0}\gamma^t b_m \in \mathbb{R}$:
\begin{align}
\overline{\sigma}^2_{i} = \frac{1}{H-1}\sum^{H-1}_{t=0}(c_m(G_{t,i} + B_m) - c_m(\mu_{i} + B_m))^2 = \frac{c_m^2}{H-1}\sum^{H-1}_{t=0}(G_{t,i} - \mu_{i})^2 = (c_m\sigma_{i})^2
\end{align}
\begin{align}
\frac{1}{c_m\sigma_{i}}(c_m(G_{t,i} + B_m) - c_m(\mu_{i} + B_m)) = \frac{1}{\sigma_{i}}(G_{t,i} - \mu_{i}) = \overline{G}_{t,i}
\end{align}
\\[4pt]
Therefore, any reward change to $G_{t,i}$ via scaling and shifting does not affect the normalized $\overline{G}_{t,i}$.
\end{proof}

According to Lemma \ref{lemma:scaling_invariance}, the non-PI baselines \textit{Naive Learning}, \textit{LOLA-PG}, and \textit{POLA-DiCE} are not affected by the dynamic reward changes in Section \ref{subsec:results_drift} and Fig. \ref{fig:drift_functions} (but are still less cooperative than \textit{DRIVE}).
However, most state-of-the-art PI methods like \textit{LIO} and \textit{MATE} cannot accommodate such changes (by the generally unknown function $f_{\textit{mod}}$) in their payoff modifications (Fig. \ref{fig:payoff_tables_pi}) and thus fail to adapt, in contrast to \textit{DRIVE}, as shown in Figs. \ref{fig:intro_figure} and \ref{fig:drift_functions}.

\subsection{Invariance in SSDs Under Normalized Policy Gradients}
\label{app:invariance_SSD}
We can restate and prove Theorem~\ref{theorem:DRIVE_invariance_SSDs} in more detail using the normalized return defined in \autoref{lemma:scaling_invariance}
$$
\bar{G}_{t,i} = \frac{G_{t,i} - \mu_i}{\sigma_i},
$$
where $\mu_i$ and $\sigma_i$ are the mean and standard deviation of all returns $G_{t,i}$ for agent $i$ in epoch $m$.

\begin{lemma}[Scaling-invariance of normalized policy gradients]
Let rewards in epoch $m$ be transformed by $r'_{t,i} = c_m r_{t,i} + b_m$ with $c_m>0$, and let $G'_{t,i}$, $\mu'_i$, and $\sigma'_i$ be the corresponding returns, mean, and standard deviation. Then the normalized
returns coincide, $\bar{G}'_{t,i} = \bar{G}_{t,i}$, and thus the policy-gradient estimate in Eq.~\eqref{eq:policy_gradients} is unchanged.
\end{lemma}

\begin{proof}
As shown in Appendix~\ref{subsec:appendix_normalization}, the transformed returns satisfy $G'_{t,i} = c_m G_{t,i} + B_m$ with $B_m = \sum_{k=0}^{H-1} \gamma^k b_m$ independent of $t$.
Therefore $\mu'_i = c_m \mu_i + B_m$ and $\sigma'_i = c_m \sigma_i$. The normalized return becomes
$$
\bar{G}'_{t,i}=\frac{G'_{t,i} - \mu'_i}{\sigma'_i}=\frac{c_m G_{t,i} + B_m - (c_m \mu_i + B_m)}{c_m \sigma_i}=\frac{G_{t,i} - \mu_i}{\sigma_i}=\bar{G}_{t,i}.
$$
Substituting $\bar{G}'_{t,i}$ into the policy gradient estimator Eq.~\eqref{eq:policy_gradients} yields the same gradient as for $\bar{G}_{t,i}$.
\end{proof}

We now combine this with the structure of DRIVE.

\begin{proof}[Proof of \autoref{theorem:DRIVE_invariance_SSDs}.]
Let the environmental rewards in epoch $m$ be transformed by $f_{\textit{mod}}(r_{t,i},m)=c_m r_{t,i}+b_m$ with $c_m>0$. 
For any agent $i$, the DRIVE shaping rule in Eq.~\eqref{eq:DRIVE_reward_app} depends only on $\hat{u}_{t,i}$ and differences of the form $\Delta_{t,i,j} = \bar{u}_j - \hat{u}_{t,i}$. 
Under $f_{\textit{mod}}$ these become
$$
\Delta'_{t,i,j}=\bar{u}'_j - \hat{u}'_{t,i}=(c_m \bar{u}_j + b_m) - (c_m \hat{u}_{t,i} + b_m)=c_m (\bar{u}_j - \hat{u}_{t,i})=c_m \Delta_{t,i,j}.
$$
The minima in Eq.~\eqref{eq:DRIVE_reward_app} therefore scale by the same factor $c_m$, and the shaped reward transforms as 
$$
u^{\mathrm{DRIVE}\,'}_{t,i}=c_m u^{\mathrm{DRIVE}}_{t,i} + b_m,
$$
i.e., $u^{\mathrm{DRIVE}}_{t,i}$ is itself subjected to the same per-epoch positive affine map. 
Applying Lemma~\ref{lemma:scaling_invariance} to $u^{\mathrm{DRIVE}}_{t,i}$ instead of $u_{t,i}$ shows that the normalized returns (and thus the policy updates) are unchanged.
Hence both the baseline policy-gradient dynamics and the incentive effects induced by DRIVE are invariant to per-epoch positive affine transformations, as claimed.
\end{proof}

\section{Technical Details and Hyperparameters}\label{app:appendix_technical}

\subsection{Hyperparameters}\label{subsec:hyperparameters}

All common hyperparameters used by all MARL approaches in the experiments, as reported in Section \ref{sec:results} in the paper, are listed in Table \ref{tab:common_hyperparameters}. The final values were chosen based on a coarse grid search which finds a tradeoff between performance and computation w.r.t. \textit{DRIVE}, \textit{LIO}, \textit{IA}, and \textit{MATE} in \textit{IPD} and \textit{Coin-2} without dynamically changing rewards, i.e., $f_{\textit{mod}}(u,m) = u$. We directly adopted the final values in Table \ref{tab:common_hyperparameters} for all other domains. 
All training and test runs were performed in parallel on a computing cluster of 15 x86\_64 GNU/Linux (Ubuntu 18.04.5 LTS) machines with i7-8700 @ 3.2GHz CPU (8 cores) and 64 GB RAM. We did not use any GPU in our experiments.

\begin{table*}[!ht]
\centering
\caption{Common hyperparameters and their respective final values used by all algorithms evaluated in the paper. We also list the numbers that have been tried during the development of the paper.}
\begin{tabular}{|L{1.75cm}|L{2.75cm}|L{2.85cm}|L{8.5cm}|} \hline
Hyperparam. & Final Value & Numbers/Range & Description\\ \hline\hline
$K$ & $10$ & \{$1$, $5$, $10$, $20$\} & Number of episodes per epoch.\\\hline
$E$ & $4000$ & \{$2000$, $4000$, $8000$\} & Number of epochs. $E$ was gradually increased to assess the stability of the learning progress until convergence. \\\hline
$\alpha$ & $0.001$ & \{$0.001$\} & Learning rate. We used the default value of ADAM in \texttt{torch} without further tuning.\\\hline
Clip norm & $1$ & \{$1$,$\infty$\} & Gradient clipping parameter. Using a clip norm of $1$ leads to better performance than disabling it with $\infty$.\\\hline
$\lambda$ & $1$ & \{$0$, $1$\} & Trace parameter for TD($\lambda$) critic learning.\\\hline
$\gamma$ & $0.95$ (\textit{IPD}, \textit{Coin-2}, \textit{Coin-4}) $0.99$ (\textit{Harvest-12}) & \{$0.9$, $0.95$, $0.99$\} & Discount factor for the return $G_{t,i}$. Any value $\geq 0.95$ would be sufficient.\\ \hline
$|\tau_{t,i}|$ & $1$ & \{$1$, $5$, $10$\} & Local history length. It was set to 1 to reduce computation because the other values did not significantly improve performance.\\ \hline
\end{tabular}\label{tab:common_hyperparameters}
\end{table*}

For \textit{LIO}, we set the cost weight for incentive function learning to 0.001 and the maximum incentive value $R_{\textit{max}}$ to the highest absolute penalty per domain (3 in \textit{IPD}, 2 in \textit{Coin-2} and \textit{Coin-4}, and 0.25 in \textit{Harvest-12}), as suggested in \cite{yang2020learning}.
For \textit{IA}, we set $\alpha = 5$ and $\beta = 0.05$, as suggested in \cite{hughes2018inequity}.

\subsection{Neural Network Architectures}
We coarsely tuned the neural network architectures in the paper w.r.t. performance and computation by varying the number of units per hidden layer \{$32$, $64$, $128$\} for $\hat{\pi}_{i}$ and $\hat{V}_{i}$. The number of hidden layers was varied between 1, 2, and 3, but significantly deeper architectures led to deteriorating performance (possibly due to vanishing gradients). Using ELU or ReLU activation did not make any significant difference for any neural network. Thus, we stick to ELU throughout the experiments. \textit{DRIVE}, \textit{MATE}, and \textit{Naive Learning} only use $\hat{\pi}_{i}$ and $\hat{V}_{i}$ as neural networks.
The incentive function network of \textit{LIO} has the same hidden layer architecture as $\hat{\pi}_{i}$ and $\hat{V}_{i}$. In addition, the joint action of the $n-1$ other agents is concatenated to the flattened observations before being input into the incentive function, which outputs an $n-1$ dimensional vector. The output vector is passed through a sigmoid function and multiplied with $R_{\textit{max}}$ (Section \ref{subsec:hyperparameters}) afterward.

\subsection{Coin-2 and Coin-4}
We adopted the setup of \cite{foerster2018learning} in \textit{Coin-2}, as shown in Fig. \ref{fig:coin-environment1}, with the same rules and reward functions. The order of executed actions is randomized such that situations where two agents simultaneously step on a blue coin lead to an expected payoff of +1 for the red agent and -1 for the blue agent (Fig. \ref{fig:coin-environment1} left), and vice versa for a red coin. In addition, we extended the domain to 4 agents in \textit{Coin-4} (Fig. \ref{fig:coin-environment1} right).
All agents are able to move freely, and grid cell positions can be occupied by multiple agents. Any attempt to move out of bounds is treated as ``do nothing" action.
\\[4pt]
\textbf{Prisoner's Dilemma Connection}
Assuming that an agent always picks up a coin at every time step in \textit{Coin-2}, the expected rewards per time step represent a PD instance, according to Section \ref{subsec:social_dilemmas} and Fig. \ref{fig:social_dilemma_matrix}:
\begin{itemize}
\item If both agents cooperate, i.e., only collect their matching coins, each agent gets $R = 0.5(1 + 0) = 0.5$ on average.
\item If both agents defect, i.e., collect any coin regardless of their color, each agent gets $P = 0.25(1 + 1 - 2 + 0) = 0$ on average.
\item If one agent defects, the defecting agent gets $T = 0.5(1 + 0) + 1 = 1.5$ on average (its own coin 100\% of the time and the other's coin 50\% of the time). The exploited agent gets $S = 0.5(1 - 2) = -0.5$
\end{itemize}
The expected payoffs $R$, $P$, $T$, and $S$ satisfy the characteristic PD inequalities w.r.t. greed and fear, namely $T > R > P > S$. Thus, DRIVE also works well in \textit{Coin-2} and \textit{Coin-4} (which extends the SD to 4 agents), even with changing rewards, as shown in Fig. \ref{fig:balance_results_star_drift}.

\begin{figure}[ht]\centering
\subfloat[Coin Domain\label{fig:coin-environment1}]{\includegraphics[width=0.4\textwidth]{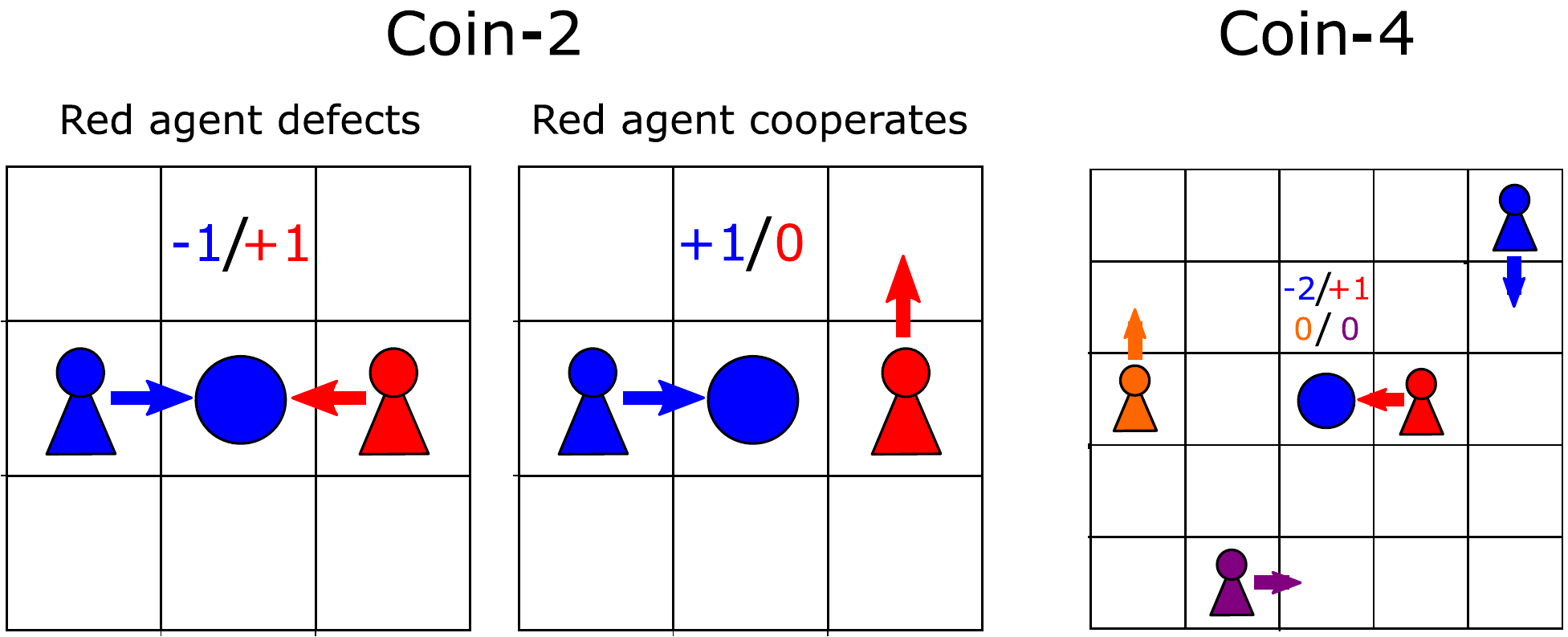}}\hspace{2em}
\subfloat[Payoff Table for \textit{Coin-2} (per time step)\label{fig:coin-payoffs}]{
\begin{tabular}[b]{|P{4em}|P{5em}|P{5em}|}
\hline & $C$ & $D$ \\ \hline 
$C$ & $0.5$, $0.5$ & $-0.5$, $1.5$ \\ \hline 
$D$ & $1.5$, $-0.5$ & $0$, $0$ \\ \hline
\end{tabular}}
\caption{\textbf{(a)} \textit{Coin-2} and \textit{Coin-4} as used in the paper. \textbf{(b)} Payoff table for \textit{Coin-2} w.r.t. the expected rewards per time step.}\label{fig:coin_domain_appendix}
\Description{\textbf{(a)} \textit{Coin-2} and \textit{Coin-4} as used in the paper. \textbf{(b)} Payoff table for \textit{Coin-2} w.r.t. the expected rewards per time step.}
\end{figure}

\subsection{Harvest-12}
We adopted the setup of \cite{perolat2017multi} in \textit{Harvest-12}, as shown in Fig. \ref{fig:harvest-environment1}, with the same dynamics and apple regrowth rates. \textit{Harvest-12} always starts with the initial apple configuration in Fig. \ref{fig:harvest-environment1} with randomly positioned agents.
The agents are able to observe the environment around their $7 \times 7$ area and have no specific orientation. Thus, they have 4 separate actions to tag all neighbor agents, which are either north, south, west, or east of them.
All agents are able to move freely, and grid cell positions can be occupied by multiple agents. Any attempt to move out of bounds is treated as ``do nothing" action.
\\[4pt]
\textbf{Prisoner's Dilemma Connection }
We empirically tested a 2-agent instance of Harvest to determine the payoff table, similar to \cite{leibo2017multi}. The table is shown in Fig. \ref{fig:harvest-payoffs} and represents a PD instance, according to Section \ref{subsec:social_dilemmas} and Fig. \ref{fig:social_dilemma_matrix}:
\begin{itemize}
\item If both agents cooperate, i.e., no stunning and occasional waiting for apple regrowth, each agent gets $R \approx 0.3$ on average.
\item If both agents defect, i.e., stunning and complete harvesting that prevents apple regrowth, each agent gets $P \approx 0.06$ on average.
\item If one agent defects, i.e., stunning and harvesting, the defecting agent gets $T \approx 0.9$, and the exploited agent gets $S \approx -0.01$ on average (due to being stunned most of the time with a 0.01 time penalty).
\end{itemize}
The expected payoffs $R$, $P$, $T$, and $S$ satisfy the characteristic PD inequalities w.r.t. greed and fear, namely $T > R > P > S$. Thus, DRIVE also works well in \textit{Harvest-12} (which extends the SD to 12 agents), even with changing rewards, as shown in Fig. ~\ref{fig:balance_results_star_drift} and Fig .~\ref{fig:star_results_alternatives2}. This also confirms the empirical result of \cite{leibo2017multi}, stating that Harvest can represent a PD instance w.r.t. greed and fear, depending on the availability and regrowth rate of apples.

\begin{figure}[ht]\centering
\subfloat[Harvest Domain\label{fig:harvest-environment1}]{\includegraphics[width=0.4\textwidth]{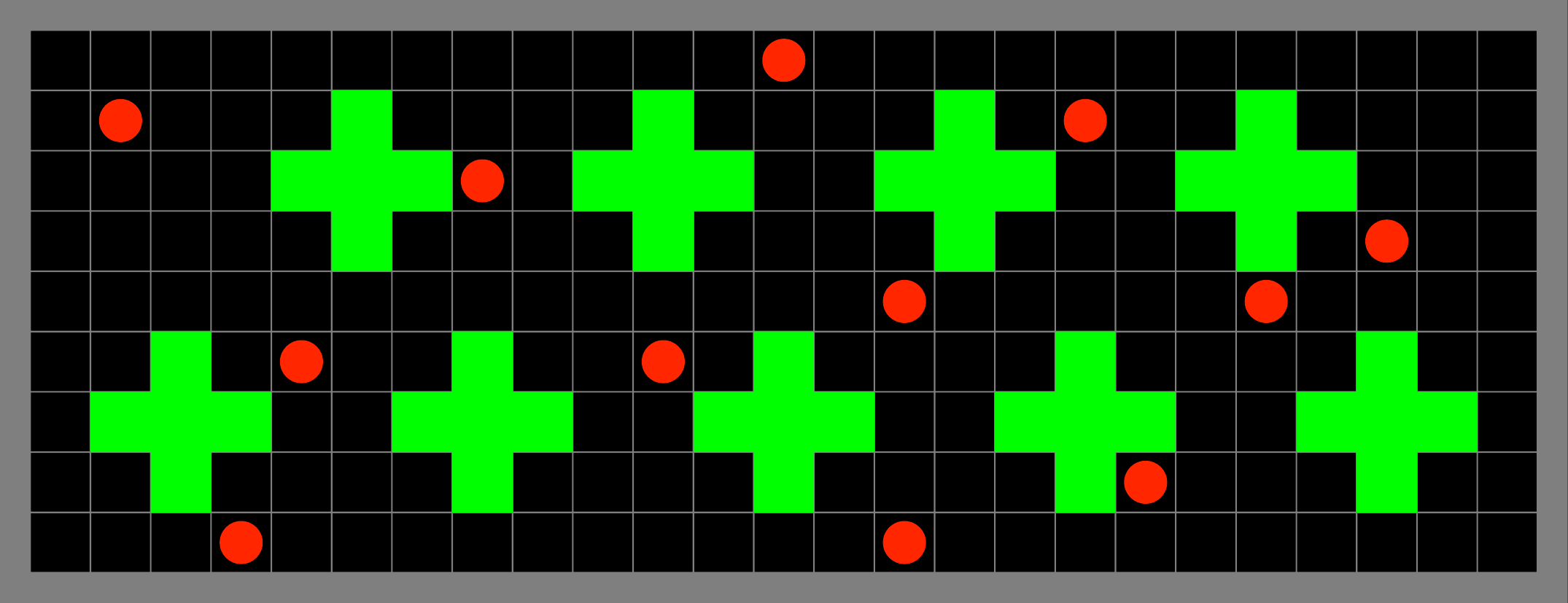}}
\hspace{2em}
\subfloat[Payoff Table for \textit{Harvest-2} (per time step)\label{fig:harvest-payoffs}]{
\begin{tabular}[b]{|P{4em}|P{5em}|P{5em}|}
\hline & $C$ & $D$ \\ \hline 
$C$ & $0.3$, $0.3$ & $-0.01$, $0.9$ \\ \hline 
$D$ & $0.9$, $-0.01$ & $0.06$, $0.06$ \\ \hline
\end{tabular}}
\caption{\textbf{(a)} Domain layout with the initial apple configuration used for \textit{Harvest-12}. \textbf{(b)} Payoff table for \textit{Harvest-2} w.r.t. the expected rewards per time step.}\label{fig:harvest_domain_appendix}
\Description{\textbf{(a)} Domain layout with the initial apple configuration used for \textit{Harvest-12}. \textbf{(b)} Payoff table for \textit{Harvest-2} w.r.t. the expected rewards per time step.}
\end{figure}

\section{Additional Results}\label{sec:additional_results}
We provide additional results regarding different cooperation measures for \textit{Harvest-12}, namely \emph{Social Welfare or Efficiency (U)} (according to Section \ref{subsec:problem_formulation}), \emph{Equality (E)} (1 minus the Gini coefficient),  \emph{Sustainability (S)} (the average number of time steps at which apples are collected), and \emph{Peace (P)} (the average number of untagged agents at any time step).
All measures are based on the \emph{original environmental rewards} $u_{t,i}$ instead of $\hat{u}_{t,i} = f_{\textit{mod}}(u_{t,i}, m)$ (Algorithm \ref{algorithm:MARL_DRIVE}, Line 15) to reliably assess the stability of cooperation.
They are defined by \cite{perolat2017learning}:

\begin{align}
% U &= \sum_{i \in \mathcal{D}}\sum_{t=0}^{H-1} u_{t,i}, \\
E &= 1 - \frac{\sum_{i \in \mathcal{D}}\sum_{j \in \mathcal{D}}|\sum_{t=0}^{H-1} (u_{t,i}-u_{t,j})|}{2n\sum_{i \in \mathcal{D}}\sum_{t=0}^{H-1} u_{t,i}}, \\
S &= \frac{1}{n}\sum_{i \in \mathcal{D}}\chi_{i} \text{, where } \chi_i = \mathbb{E}[t| u_{t,i} > 0], \\
P &= n -  \frac{1}{H}\sum_{i \in \mathcal{D}}\sum_{t = 1}^{H}\mathbf{1}[\text{agent timed-out on time step }t]
\end{align}

Fig. \ref{fig:star_results_alternatives2} shows the progress of the alternative cooperation measures in \textit{Harvest-12} with the different reward change functions $f_{\textit{mod}}$ (I, II, III, and IV) from Section \ref{subsec:results_drift}. 
While DRIVE is robust w.r.t. all measures, MATE and LIO significantly change under most settings, except for equality regarding reward change functions $f^{I}_{\textit{mod}}$ and $f^{III}_{\textit{mod}}$. 
IA is only affected by $f^{IV}_{\textit{mod}}$ regarding all cooperation measures. 
Apart from $f^{III}_{\textit{mod}}$, MATE always exhibits more peace than DRIVE, although the peace level significantly varies depending on the reward change function. 
LIO only exhibits more peace than DRIVE under $f^{II}_{\textit{mod}}$ and $f^{IV}_{\textit{mod}}$, which both converge to zero over time.

\begin{figure}[htb]\centering
\includegraphics[width=0.82\textwidth]{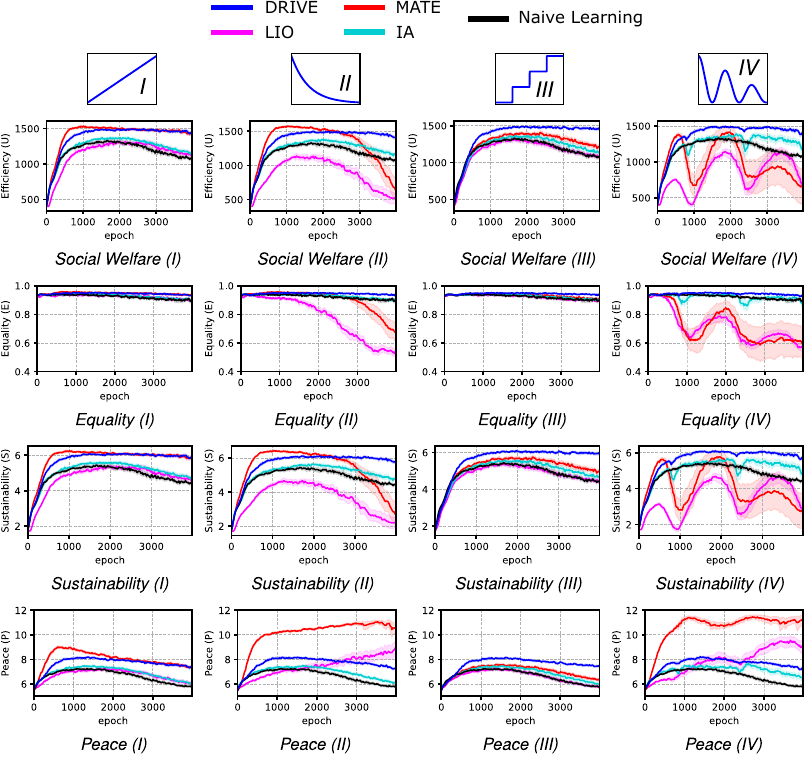}
\caption{Average progress of DRIVE and other baselines in \textit{Harvest-12} with different reward change functions $f_{\textit{mod}}$ (I, II, III, and IV), as illustrated in the legend at the top and Fig. \ref{fig:drift_functions} regarding the cooperation measures \emph{Social Welfare or Efficiency (U)}, \emph{Equality (E)},  \emph{Sustainability (S)}, and \emph{Peace (P)} \cite{perolat2017learning}. Shaded areas show the 95\% confidence interval.}\label{fig:star_results_alternatives2}
\Description{Average progress of DRIVE and other baselines in \textit{Harvest-12} with different reward change functions $f_{\textit{mod}}$ (I, II, III, and IV), as illustrated in the legend at the top and Fig. \ref{fig:drift_functions} regarding the cooperation measures \emph{Social Welfare or Efficiency (U)}, \emph{Equality (E)},  \emph{Sustainability (S)}, and \emph{Peace (P)} \cite{perolat2017learning}. Shaded areas show the 95\% confidence interval.}
\end{figure}

%%%%%%%%%%%%%%%%%%%%%%%%%%%%%%%%%%%%%%%%%%%%%%%%%%%%%%%%%%%%%%%%%%%%%%%%

\end{document}